\documentclass[journal]{IEEEtran}

\usepackage{cite}
\usepackage{empheq}
\usepackage{amsmath}
\usepackage{amsthm}
\usepackage{amssymb}
\usepackage{mathrsfs}
\usepackage{ulem}
\usepackage{graphicx}
\graphicspath{{Figs/}{../pdf/}{../jpeg/}}
\DeclareGraphicsExtensions{.pdf,.jpeg,.png}

\usepackage{cancel}
\usepackage{makecell}
\usepackage{bbding}
\usepackage{booktabs}
\usepackage{array}

\newtheorem{theorem}{Theorem}[section]

\newtheorem{corollary}[theorem]{Corollary}
\newtheorem{proposition}[theorem]{Proposition}

\newtheorem{assumption}[]{Assumption}
\newtheorem{example}{Example}

\usepackage{xcolor}

\begin{document}

\title{Robust Data-Enabled Predictive Control: Tractable Formulations and Performance Guarantees}
\author{Linbin Huang$^\dag$, Jianzhe Zhen$^\dag$, John Lygeros, and Florian D{\"o}rfler
\thanks{\dag: The first two authors contributed equally to this work.}
\thanks{The authors are with the Department of Information Technology and Electrical
Engineering at ETH Z{\"u}rich, Switzerland. (Emails: \text\{linhuang, jizhen, jlygeros, dorfler\}@ethz.ch)}
\thanks{This research was supported by the SNSF under NCCR Automation, the ERC under the project OCAL (grant number 787845), and ETH Zurich Funds.}
}

	
	\maketitle
	
	\begin{abstract}
	We introduce a general framework for robust data-enabled predictive control (DeePC) for linear time-invariant (LTI) systems. The proposed framework enables us to obtain model-free optimal control for LTI systems based on noisy input/output data. More specifically, robust DeePC solves a min-max optimization problem to compute the optimal control sequence that is resilient to all possible realizations of the uncertainties in the input/output data within a prescribed uncertainty set. We present computationally tractable reformulations of the min-max problem with various uncertainty sets. Furthermore, we show that even though an accurate prediction of the future behavior is unattainable in practice due to inaccessibility of the perfect input/output data, the obtained robust optimal control sequence provides performance guarantees for the actually \textit{realized} input/output cost.
	We further show that the robust DeePC generalizes and robustifies the regularized DeePC (with quadratic regularization or 1-norm regularization) proposed in the literature. Finally, we demonstrate the performance of the proposed robust DeePC algorithm on high-fidelity, nonlinear, and noisy simulations of a grid-connected power converter system.
	
	\end{abstract}
	
	
	\section{Introduction}\label{sec:intro}
	
	Data-driven control seeking an optimal control strategy from data is attracting increasing interest from both academia and industry. Compared to conventional model-based control, data-driven control has the advantage that it can be applied in scenarios where data is readily available, but the system and uncertainty models are too complex to obtain or maintain, e.g., large-scale power systems or energy efficient buildings.
	
	There are mainly two paradigms of data-driven control: 1)~\textit{indirect data-driven control} that first identifies a model and then conducts control design based on the identified model, and 2)~\textit{direct data-driven control} that circumvents the step of system identification and obtains control policies directly from data. Indirect data-driven control has a long history in control, where many methods have been developed for identifying the system model, e.g., prediction error methods (PEMs), maximum likelihood methods, and subspace methods~\cite{ljung1999system}; the subsequent optimal control can be conducted using, for instance, model predictive control (MPC), linear–quadratic–Gaussian control, and system-level synthesis~\cite{gevers2005identification,favoreel1999spc,morari1999model,pillonetto2014kernel,boczar2018finite}.
	Direct data-driven control gained increasing attention and became popular thanks to iterative feedback tuning \cite{hjalmarsson1998iterative}, virtual reference feedback tuning~\cite{campi2002virtual}, reinforcement learning~\cite{recht2019tour}, etc. The connection between indirect data-driven control and direct data-driven control was investigated in \cite{dorfler2021bridging}, which provides new insights from the perspective of regularization and convex relaxations in optimization. A central promise is that direct data-driven control may have higher flexibility and better performance than indirect data-driven control thanks to the data-centric representation that avoids using a specific model from identification.
	Moreover, it is generally difficult to map uncertainty specifications from system identification over to robust control in indirect data-driven control, while, as we will show in this paper, this may become easier in direct data-driven control.
	
	In recent years, a result originally formulated by Willems and co-authors in the context of behavioral system theory has received renewed attention in direct data-driven control. This result, known as the \textit{Fundamental Lemma}~\cite{willems2005note}, shows that the subspace of input/output trajectories of a linear time-invariant (LTI) system can be obtained from the column span of a data Hankel matrix, thereby avoiding a parametric system representation. This result was extended in~\cite{markovskyidentifiability} and~\cite{van2020willems} to consider mosaic Hankel matrices, (Chinese) Page matrices, and trajectory matrices as data-driven predictors. Recently, multiple direct data-driven control methods have been proposed based on the Fundamental Lemma, e.g.,~\cite{coulson2019data,de2019formulas,berberich2019data,van2019data,coulson2020distributionally,wieler2021data,lian2021nonlinear,alexandru2020data}.
	
	Here we concentrate on the {\bf d}ata-{\bf e}nabl{\bf e}d {\bf p}redictive {\bf c}ontrol (DeePC) proposed in~\cite{coulson2019data}. In the spirit of the Fundamental Lemma, the DeePC algorithm relies only on input/output data to learn the behavior of the unknown system and perform safe and optimal control to drive the system along a desired trajectory using real-time feedback. The DeePC algorithm has been successfully applied in many scenarios, including power systems~\cite{huang2019decentralized,huang2019data}, motor drives~\cite{carlet2020data}, and quadcopters~\cite{elokda2019data}.
	
	When perfect (noiseless and uncorrupted) input/output data is accessible, DeePC can accurately predict the future behavior of the system thanks to the Fundamental Lemma. In this case, DeePC has equivalent closed-loop behavior to conventional MPC with a model and perfect state estimation~\cite{coulson2019data}. However, in practice, perfect data is in general not accessible to the controller due to measurement noise and process noise, which leads to inaccurate estimations and predictions and may degrade the quality of the obtained optimal control sequence.
	In fact, a key question for data-driven control is: how does the system perform when applying control policies computed from noisy data? Usually, uncertainty quantification is needed in system identification and robustification is needed in the design of the control policy, e.g., through regularization~\cite{berberich2019data,coulson2020distributionally}, computation of reachable regions~\cite{Alanwar2021}, or robust system level synthesis~\cite{xue2020data}. In practice, the uncertainty descriptions of identification and robust control are often incompatible.
	For DeePC, it has been frequently observed that regularization is very important to ensure good performance under noisy measurements. Theoretical support for this observation was provided in~\cite{coulson2020distributionally} and \cite{coulson2019regularized}, where it was shown that regularization provides distributional robustness against stochastic disturbances. The authors of~\cite{berberich2019data} showed that a quadratic regularization is also essential for stability. Moreover, regularization leads to a convex relaxation of the indirect (first identify, then control) approach and accounts for an implicit identification step~\cite{dorfler2021bridging}.
	In our recent work~\cite{huang2020quad}, we showed that including a quadratic regularization in DeePC enables the reformulation as a min-max problem, which provides robustness against uncertainties in the output data. Regularization was also linked to a min-max formulation in~\cite{xue2020data} for system level synthesis.
	Nonetheless, it still remains unclear whether the realized input/output cost can be guaranteed by applying the control sequence computed from noisy data, and whether different structural assumptions on the uncertainties (e.g., with Hankel structures) can be taken into account to reduce the conservativeness induced from coarse robustification.
	
	To this end, we present a robust DeePC framework that involves solving a min-max optimization problem to robustify the control sequence against uncertainties in the input/output data. We justify our min-max formulation by showing that it enables a performance guarantee for the realized input/output cost of the system.
	We consider different uncertainty sets as tight estimates of input/output uncertainties when different types of data matrices (e.g., Hankel or Page matrices) are used as predictors. For instance, to make the considered set tight, one may incorporate Hankel structures on uncertainties when Hankel matrices are employed, and may consider column-wise uncertainties when Page matrices or trajectory matrices are used. For such uncertainty structures we explicitly show how the min-max problems can be reduced to tractable minimization problems. In particular, we discuss the connections between the robust DeePC and (previously proposed and novel) regularized DeePC algorithms, and illustrate how different uncertainty sets lead to different regularization terms, extending our result on quadratic regularization in \cite{huang2020quad}.
	
	The rest of this paper is organized as follows: in Section~II we give a brief review on the DeePC algorithm. Section~III presents the robust DeePC framework. Section~IV derives tractable formulations for robust DeePC under different uncertainty sets. In Section~V we discuss the connections between robustness and regularization. Section~VI tests the robust DeePC with high-fidelity simulations on a power converter system. We conclude the paper in Section~VII.
	
	\section*{Notation}
Let $\mathbb N$ denote the set of positive integers, $\mathbb R$ denote the set of real numbers, $\mathbb R_{\ge 0}$ the set of nonnegative real numbers, $\mathbb R^n$ the $n$-dimensional Euclidean space, $\mathbb R^{m \times n}$ the set of real $m$-by-$n$ matrices, and $\mathbb S^{n}$ $(\mathbb S^{n}_{>0})$ the set of real $n$-by-$n$ symmetric (positive definite) matrices. We denote the index set with cardinality $n\in \mathbb N$ as $[n]$, that is, $[n] = \{ 1, ..., n \}$.
We use $\|x\|$ ($\|A\|$) to denote the (induced) 2-norm of the vector $x$ (the matrix $A$); we use $\|x\|_1$ to denote the 1-norm of the vector $x$; for a vector $x$, we use ${\left\| x \right\|_A^2}$ to denote  $x^\top A x$; for a matrix $A$, we use $\|A\|_Q$ to denote $\|Q^{\frac{1}{2}}A\|$; we use $\|A\|_F$ to denote the Frobenius norm of the matrix $A$.
We use $\boldsymbol 1_n$ to denote a vector of ones of length $n$, $\boldsymbol 1_{m\times n}$ to denote a $m$-by-$n$ matrix of ones, and $I_n$ to denote an $n$-by-$n$ identity matrix (abbreviated as $I$ when the dimensions can be inferred from the context). We use $A^+$ to denote the right inverse of $A$, and $A^\bot = I - A^+A$ to denote the orthogonal projector onto the kernel of $A$. We use ${\rm col}(Z_0,Z_1,...,Z_\ell)$ to denote the matrix $[Z_0^{\top}\; Z_1^{\top}\; \cdots \; Z_\ell^{\top}]^{\top}$. We use $\otimes$ to denote the Kronecker product.

	\section{Data-Enabled Predictive Control}
	
	\subsection{Notation and Preliminaries on the Fundamental Lemma}
	
	Consider a discrete-time linear time-invariant (LTI) system
	\begin{equation}
		\left\{ \begin{array}{l}
			{x_{t + 1}} = A{x_t} + B{u_t}\\
			{y_t} = C{x_t} + D{u_t},
		\end{array} \right.\,		\label{eq:ABCD}
	\end{equation}
	where $A \in \mathbb{R}^{n \times n}$, $B \in \mathbb{R}^{n \times m}$, $C \in \mathbb{R}^{p \times n}$, $D \in \mathbb{R}^{p \times m}$, $x_t \in \mathbb{R}^n$ is the state of the system at~$t \in \mathbb{Z}_{ \ge 0}$,~$u_{t} \in \mathbb{R}^m$ is the input vector, and $y_{t} \in \mathbb{R}^p$ is the output vector. Recall the respective extended observability matrix and convolution (impulse-response) matrices
	\vspace{-0.5mm}
	\begin{equation*}
		\mathscr{O}_{\ell}(A,C) := {\rm col}(C,CA,...,CA^{\ell-1}) \,, \quad {\rm and}
		\vspace{-0.5mm}
	\end{equation*}
	\begin{equation}\label{eq:T_N}
\mathcal{T}_N = \begin{bmatrix}
D & 0 & 0 & \cdots & 0 \\
CB & D & 0 & \cdots & 0 \\
CAB & CB & D  & \cdots & 0 \\
\vdots & \vdots & \vdots & \ddots & \vdots \\
CA^{N-2}B & CA^{N-3}B & CA^{N-4}B & \cdots & D \\
\end{bmatrix} \,.
\end{equation}
	The \textit{lag} of the system (\ref{eq:ABCD}) is defined by the smallest integer $\ell \in \mathbb{Z}_{ \ge 0}$ such that the observability matrix $\mathscr{O}_{\ell}(A,C)$ has rank $n$, i.e., the state can be reconstructed from $\ell$ measurements.
	In a data-driven setting, $\ell$ and $n$ are generally unknown, but upper bounds can usually be inferred from knowledge of the system. Consider $L,T \in \mathbb{Z}_{ \ge 0}$ with $T \geq L > \ell$ and length-$T$ input and output trajectories of \eqref{eq:ABCD}: $u = {\rm col} (u_0,u_1,\dots u_{T-1})\in \mathbb{R}^{mT}$ and $y ={\rm col} (y_0,y_1, \dots y_{T-1})\in \mathbb{R}^{pT}$. For the inputs $u$, define the Hankel matrix of depth $L$ as
	\begin{equation}
		\mathscr{H}_L(u) := \left[ {\begin{array}{*{20}{c}}
				{{u_0}}&{{u_1}}& \cdots &{{u_{T - L}}}\\
				{{u_1}}&{{u_2}}& \cdots &{{u_{T - L + 1}}}\\
				\vdots & \vdots & \ddots & \vdots \\
				{{u_{L-1}}}&{{u_{L}}}& \cdots &{{u_{T-1}}}
		\end{array}} \right] \,.		
		\label{eq:Hankel_L}
	\end{equation}
	Accordingly, for the outputs define the Hankel matrix $\mathscr{H}_L(y)$. Consider the stacked matrix
	$\mathscr{H}_L(u,y) = \left[\begin{smallmatrix}\mathscr{H}_L(u)\\\mathscr{H}_L(y)\end{smallmatrix}\right]$. By \cite[Corollary 19]{markovskyidentifiability},  the restricted behavior (of length $L$) equals the image of $\mathscr{H}_L(u,y)$ if and only if rank$\left(\mathscr{H}_L(u,y)\right)=mL + n$. 
	Note that this result extends and includes the original  Fundamental Lemma  \cite[Theorem 1]{willems2005note} which requires controllability and persistency of excitation of order $L+n$ (i.e., $\mathscr{H}_{L+n}(u)$ must have full row rank) as sufficient conditions.
	
	These behavioral results can be leveraged for data-driven prediction and estimation as follows. Consider $T_{\rm ini},N,T \in \mathbb{Z}_{ \ge 0}$, as well as an input/output time series ${\rm col}(u^{\rm{d}},y^{\rm{d}}) \in \mathbb{R}^{(m+p)T}$ so that rank$\left(\mathscr{H}_{T_{\rm ini}+N}(u^{\rm{d}},y^{\rm{d}})\right)=m(T_{\rm ini}+N) + n$. Here the superscript ``d'' denotes data collected offline, and the rank condition is met by choosing $u^{\rm{d}}$ to be  persistently exciting of sufficiently high order. The Hankel matrix $\mathscr{H}_{T_{\rm ini}+N}(u^{\rm{d}},y^{\rm{d}})$ can be partitioned into
	\begin{equation*}
		\left[ {\begin{array}{*{20}{c}}
				{{\bar U_{\rm P}}}\\
				{{\bar U_{\rm F}}}
		\end{array}} \right] := \mathscr{H_c}_{T_{\rm ini}+N}(u^{\rm{d}}) \quad \text{and} \quad \left[ {\begin{array}{*{20}{c}}
				{{\bar Y_{\rm P}}}\\
				{{\bar Y_{\rm F}}}
		\end{array}} \right] := \mathscr{H_c}_{T_{\rm ini}+N}(y^{\rm{d}})\,,		\label{eq:partition_Huy}
	\end{equation*}
	where $\bar U_{\rm P} \in \mathbb{R}^{mT_{\rm ini} \times H_c}$, $\bar U_{\rm F} \in \mathbb{R}^{mN \times H_c}$, $\bar Y_{\rm P} \in \mathbb{R}^{pT_{\rm ini} \times H_c}$, $\bar Y_{\rm F} \in \mathbb{R}^{pN \times H_c}$, and $H_c = T-T_{\rm ini}-N+1$. In the sequel, the data in the partition with subscript P (for ``past'') will be used to implicitly estimate the initial condition of the system, whereas the data with subscript F will be used to predict the ``future'' trajectories. In this case, $T_{\rm ini}$ is the length of an initial trajectory measured in the immediate past during on-line operation, and $N$ is the length of a predicted trajectory starting from the initial trajectory. Recall that the image of $\mathscr{H}_{T_{\rm ini}+N}(u^{\rm{d}},y^{\rm{d}})$ spans all length-$(T_{\rm ini}+N)$ trajectories, that is,
 ${\rm{col}}(u_{\rm ini},u,y_{\rm ini},y) \in \mathbb R^{(m+p)(T_{\rm ini}+N)}$ is a trajectory of (\ref{eq:ABCD}) if and only if there exists $g \in \mathbb{R}^{H_c}$ so\,that
\begin{equation}\label{eq:Hankel_g}
\left[ {\begin{array}{*{20}{c}}
	{{\bar U_{\rm P}}}\\
	{{\bar Y_{\rm P}}}\\
	{{\bar U_{\rm F}}}\\
	{{\bar Y_{\rm F}}}
	\end{array}} \right]g = \left[ {\begin{array}{*{10}{c}}
	{{\bar u_{\rm ini}}}\\
	{{\bar y_{\rm ini}}}\\
	u\\
	y
	\end{array}} \right]\,.		
\end{equation}
The initial trajectory ${\rm{col}}(\bar u_{\rm ini},\bar y_{\rm ini}) \in \mathbb R^{(m+p)T_{\rm ini}}$ can be thought of as setting the initial condition for the future (to be predicted) trajectory ${\rm{col}}(u,y)\in \mathbb R^{(m+p)N}$. In particular, if $T_{\rm ini} \ge \ell$, for every given future input trajectory $u$, the future output trajectory $y$ is uniquely determined through (\ref{eq:Hankel_g}) \cite{markovsky2008data}.

In addition to Hankel matrices, one can also use more input/output data to construct (Chinese) Page matrices or trajectory matrices as data-driven predictors \cite{markovskyidentifiability}, where the Page matrix (of depth $L$) for a signal $u \in \mathbb{R}^{mLT'}$ is
\begin{equation}
		\mathscr{P}_L(u) := \left[ {\begin{array}{*{20}{c}}
				{{u_0}}&{{u_L}}& \cdots &{{u_{L(T' - 1)}}}\\
				{{u_1}}&{{u_{L+1}}}& \cdots &{{u_{L(T' - 1)+1}}}\\
				\vdots & \vdots & \ddots & \vdots \\
				{{u_{L-1}}}&{{u_{2L-1}}}& \cdots &{{u_{LT'-1}}}
		\end{array}} \right] \,,		
		\label{eq:Page_L}
\end{equation}
and a trajectory matrix (of depth $L$) can be constructed from $T'$ independent trajectories $u^i \in \mathbb{R}^{mL} $ ($i \in [T']$) as
\begin{equation}
		\mathscr{T}_L(u) := \left[ {\begin{array}{*{20}{c}}
				{{u^1}}&{{u^2}}& \cdots &{{u^{T'}}}
		\end{array}} \right] \,.		
		\label{eq:Trajectory_L}
\end{equation}
Notice that the elements in a Hankel matrix are correlated (structured), while the elements in a Page matrix or a trajectory matrix are independent (unstructured).

	\subsection{Review of the DeePC algorithm}
	
	The DeePC algorithm proposed in \cite{coulson2019data} directly uses input/output data collected from the unknown system to predict the future behaviour, and performs optimal and safe control without identifying a parametric system representation. More specifically, DeePC solves the following optimization problem to obtain the optimal future control inputs
	\begin{equation} 			\label{eq:DeePC}
		\begin{array}{cl}
			\displaystyle \mathop {{\rm{min}}}\limits_{g \atop (u, y) \in \mathcal C }  \left\{ {\left\| u \right\|_R^2} + {\left\| {y - r} \right\|_Q^2}  \, \Bigg|  \,\,   {\begin{bmatrix}
	{{\bar U_{\rm P}}}\\
	{{\bar Y_{\rm P}}}\\
	{{\bar U_{\rm F}}}\\
	{{\bar Y_{\rm F}}}
	\end{bmatrix}}g =  {\begin{bmatrix}
	{{\bar u_{\rm ini}}}\\
	{{\bar y_{\rm ini}}}\\
	u\\
	y
	\end{bmatrix}}  \right\},
		\end{array}
	\end{equation}
	where the set of input/output constraints is defined as $\mathcal C = \{ (u,y) \in \mathbb{R}^{(m+p)N} \ | \ W {\rm col}(u,y)  \le w \}$ for $W\in \mathbb R^{n_w \times (m+p)N}$ and~$w \in \mathbb R^{n_w}$. The positive definite matrix~$R \in \mathbb{S}^{mN \times mN}_{>0}$ and positive semi-definite matrix~$Q \in \mathbb{S}^{pN \times pN}_{\ge 0}$ are the cost matrices. The vector $r \in \mathbb{R}^{pN}$ is a prescribed reference trajectory for the future outputs.
	DeePC involves solving the convex quadratic problem (\ref{eq:DeePC}) in a receding horizon fashion, that is, after calculating the optimal control sequence $u^\star$, we apply $(u_t,...,u_{t+k-1}) = (u_0^{\star},...,u_{k-1}^{\star})$ to the system for $k \le N$ time steps, then, reinitialize the problem (\ref{eq:DeePC}) by updating ${\rm col}(\bar u_{\rm ini},\bar y_{\rm ini})$ to the most recent input and output measurements, and setting $t$ to $t+k$, to calculate the new optimal control for the next $k \le N$ time steps. As in MPC, the control horizon $k$ is a design parameter.

	The standard DeePC algorithm in \eqref{eq:DeePC} assumes perfect (noiseless and uncorrupted) input/output data generated from the unknown system \eqref{eq:ABCD}.
	However, in practice, the perfect data is not accessible to the controller due to, for example, measurement noise, process noise, or noise that enters through the input channels. Instead, the corresponding measured output data and the corresponding recorded input data are used in the control algorithm.
	
	{\textbf{Perfect Data and Noisy Data.}} Throughout the paper, we use $\bar U_{\rm P}$, $\bar U_{\rm F}$, $\bar u_{\rm ini}$, $\bar Y_{\rm P}$, $\bar Y_{\rm F}$, and $\bar y_{\rm ini}$ to denote the perfect data generated from the system \eqref{eq:ABCD}, which accurately captures the system dynamics according to \eqref{eq:Hankel_g}; we use $\hat U_{\rm P}$, $\hat U_{\rm F}$, and $\hat u_{\rm ini}$ to denote the corresponding recorded input data, and use $\hat Y_{\rm P}$, $\hat Y_{\rm F}$, and $\hat y_{\rm ini}$ to denote the corresponding measured output data.

	\section{Robust DeePC}
	
	In this section, we propose a novel framework for robust DeePC to incorporate different geometry and tighter estimates of uncertainty sets in the input/output data.
	
	\subsection{Robust DeePC} \label{TR_Robust_DPC}

The perfect input/output data generated from the unknown system~\eqref{eq:ABCD} precisely captures the system dynamics and accurately predicts the future behavior of the system thanks to Fundamental Lemma \cite{willems2005note}. However, in practice, the perfect data is not accessible to the controller due to, for instance, process noise, or measurement noise. This results in inaccurate prediction and may degrade performance once the control sequence is applied to the system. To deal with this problem, we propose a robust optimization framework that incorporates uncertainties in the data matrices and the initial trajectory of~\eqref{eq:DeePC} by considering the following robust counterpart
	\begin{equation}
		\begin{array}{l}
			\displaystyle \mathop {{\rm{min}}}\limits_{ g, \sigma_u, \sigma_y \atop (u, y) \in \mathcal C}  {\left\| u \right\|_R^2} + {\left\| {y - r} \right\|_Q^2}  + h(\sigma_u, \sigma_y)\\
			\;\;{\rm s.t.}\;\; {\begin{bmatrix}
					{{U_{\rm P} (\xi)}}\\
					{{Y_{\rm P} (\xi)}}\\
					{{U_{\rm F} (\xi)}}\\
					{{Y_{\rm F} (\xi)}}
			\end{bmatrix}}g = {\begin{bmatrix}
					{{u_{\rm ini} (\xi)}}\\
					{{y_{\rm ini} (\xi) }}\\
					u\\
					y
			\end{bmatrix}} + {\begin{bmatrix}
					\sigma_u \\
					\sigma_y \\
					0\\
					0
			\end{bmatrix}},\ \forall \xi \in \mathcal D
			\label{eq:RobustDeePC}
		\end{array}
	\end{equation}
	where the prescribed uncertainty set $\mathcal D \subset \mathbb R^{n_\xi}$ is a compact and convex set with nonempty relative interior, and $(\sigma_u, \sigma_y)$ captures the violation of the first two sets of equality constraints. The estimation regularizer $h(\sigma_u, \sigma_y)$ penalizes violation of the first two sets of equality constraints due to uncertainties, which can be seen as a penalty on the estimation uncertainty. We assume that $h:\mathbb R^{mN} \times \mathbb R^{pN} \rightarrow \mathbb R_{\ge 0}$ is a convex function with $h(\sigma_u, \sigma_y) = 0$ if and only if $ \sigma_u = 0$ and $\sigma_y =0$.
	
	We assume that the elements of the Hankel matrices and the initial trajectory are affected by uncertainties in an affine way. More specifically, the matrices ${{U_{\rm P} (\xi)}}, {{Y_{\rm P} (\xi)}}, {{U_{\rm F} (\xi)}}, {{Y_{\rm F} (\xi)}}$, and vectors ${{u_{\rm ini} (\xi)}}$ and ${{y_{\rm ini} (\xi) }}$ are affine functions of the uncertain parameter $\xi\in \mathcal{D}$, with ${{U_{\rm P} (0)}} = \hat U_{\rm P}, {{Y_{\rm P} (0)}} = \hat Y_{\rm P}, {{U_{\rm F} (0)}} = \hat U_{\rm F}, {{Y_{\rm F} (0)}} = \hat Y_{\rm F}$, ${{u_{\rm ini} (0)}} = \hat u_{\rm ini}$ and ${{y_{\rm ini} (0) }} = \hat y_{\rm ini}$.
	For instance,
	\[
	Y_{\rm P} (\xi) =  \hat Y_{\rm P} + \sum_{j= 1}^{n_\xi} Y_{\rm P}^{(j)} \xi_j  \quad \text{and} \quad y_{\rm ini} (\xi) =  \hat y_{\rm ini} + \sum_{j= 1}^{n_\xi} y^{(j)}_{\rm ini} \xi_j,
	\]
	and $Y_{\rm P}^{(j)} \in \mathbb{R}^{pT_{\rm ini} \times H_c}$, $y^{(j)}_{\rm ini} \in \mathbb{R}^{pT_{\rm ini}}$ for every $j \in [n_\xi]$.

	
In problem~\eqref{eq:RobustDeePC}, one seeks a vector~$(u,y,g,\sigma_u, \sigma_y)$ that minimizes the objective function value and is feasible for all possible uncertainties $\xi$ residing in $\mathcal D$.

Problem~\eqref{eq:RobustDeePC} is overly conservative, or even infeasible because there may not exist a~$(u,y,g, \sigma_u, \sigma_y)$ that satisfies the equality constraints for all possible realization of $\xi$ in $\mathcal D$.
To reduce the conservativeness of~\eqref{eq:RobustDeePC}, we propose the following two-stage robust DeePC decision problem
	\begin{equation}
		\begin{array}{l}
			\displaystyle \mathop {{\rm{min}}}\limits_{g} \max_{\xi \in \mathcal D} \min_{(u, y) \in \mathcal C \atop \sigma_u, \sigma_y }  {\left\| u \right\|_R^2} + {\left\| {y - r} \right\|_Q^2} + h(\sigma_u, \sigma_y)\\
			\;{\rm s.t.}\;\; {\begin{bmatrix}
					{{U_{\rm P} (\xi)}}\\
					{{Y_{\rm P} (\xi)}}\\
					{{U_{\rm F} (\xi)}}\\
					{{Y_{\rm F} (\xi)}}
			\end{bmatrix}} g = {\begin{bmatrix}
					{{u_{\rm ini} (\xi)}}\\
					{{y_{\rm ini} (\xi) }}\\
					u\\
					y
			\end{bmatrix}} + {\begin{bmatrix}
					\sigma_u \\
					\sigma_y \\
					0\\
					0
			\end{bmatrix}} .
			\label{eq:RobustDeePC2}
		\end{array}
	\end{equation}
	Here, we assume $g \in \mathbb R^{H_c}$ to be a first stage variable that is decided \textit{before} the realization of the disturbances $\xi$, and~$(u,y,\sigma_u, \sigma_y)$ to be the second stage variable which is determined \textit{after} the value of $\xi$ is revealed. Therefore, the second stage variable $(u,y,\sigma_u, \sigma_y)$ is a general function of~$\xi$ that minimizes the objective function after $\xi$ is realized. Problem~\eqref{eq:RobustDeePC2} should be read as: the designer chooses $g$, then nature strikes and chooses an adversarial $\xi$. In the end, the control sequence $(u,y)$ and estimation errors $(\sigma_u,\sigma_y)$ realize.
	We will later show that this two-stage problem provides robust control input sequence to the system and guarantees a performance certificate for the realized input/output cost.
	
	The coefficient of the second stage variable $(u,y,\sigma_u, \sigma_y)$ is constant in~\eqref{eq:RobustDeePC2}, which corresponds to the stochastic programming format known as \textit{fixed recourse} \cite{bl11}. In this case, since~$(u,y,\sigma_u, \sigma_y)$ are fully characterized by the uncertain linear equality constraints in~\eqref{eq:RobustDeePC2}, linear decision rules are optimal for~$(u,y,\sigma_u, \sigma_y)$. Without loss of generality, one can eliminate all the second stage variables~$(u,y,\sigma_u, \sigma_y)$ in~\eqref{eq:RobustDeePC2}, which is equivalent to imposing linear decision rules on $(u,y,\sigma_u, \sigma_y)$ \cite[Lemma 2]{zhen2018computing}, and problem~\eqref{eq:RobustDeePC2} can be equivalently reformulated as the following min-max robust optimization problem
	\begin{equation}
		\begin{array}{cl}
			\displaystyle \mathop {{\rm{min}}}\limits_{g \in \mathcal G} \max_{ \xi \in \mathcal D} & \left\| U_{\rm F} (\xi) g \right\|_R^2 + \left\| { Y_{\rm F} (\xi)g  - r} \right\|_Q^2\\
			& + h(U_{\rm P} (\xi) g  - u_{\rm ini} (\xi), Y_{\rm P} (\xi)g - y_{\rm ini} (\xi)),
			\label{eq:minmaxDeePC-penalty}
		\end{array}
	\end{equation}
	where $\mathcal G = \{ g  \ | \ \forall \xi \in \mathcal D:  \,  (U_{\rm F} (\xi) g , Y_{\rm F} (\xi)g) \in \mathcal C \}$.
		\begin{table}
		\centering
		\caption{Tractable robust counterparts for $\forall \xi \in \mathcal D: (a + B^\top \xi)^\top x  \le b$. LO: linear optimization problems. CQO: convex quadratic optimization problems. Note that $\nu$ is an additional optimization variable in the resulting tractable robust counterparts. This table is adapted from~\cite{bdv15}.}
		\label{tab:g.examples}
		\begin{tabular}{c| cc}
			\hline
			\makecell{Type\\(Tractability)}  & $\mathcal D$ & Robust Counterpart \rule{0pt}{4ex}\rule[-2.6ex]{0pt}{0pt} \\
			\hline
			\makecell{Box\\ (LO)} & $\| \xi \|_\infty \le \rho$ & $a^\top x + \rho \|B x\|_1 \le b $ \rule{0pt}{4ex}\rule[-2.6ex]{0pt}{0pt} \\ \hline
			\makecell{Ellipsoidal\\ (CQO)} & $\| \xi\| \le \rho $ & $a^\top x + \rho \|B x\| \le b $ \rule{0pt}{4ex}\rule[-2.6ex]{0pt}{0pt} \\  \hline
			\makecell{Budget\\ (LO)} 
			& $D\xi \le d $ & $\begin{cases}
				a^\top x + d^\top \nu \le b \\
				D^\top \nu = Bx \\
				\nu \ge 0
			\end{cases} $ \\ \hline
			\makecell{Polyhedral\\ (LO)} 
			& $\begin{cases}
				\| \xi \|_\infty \le \rho  \\
			 \| \xi \|_1 \le \tau
			\end{cases}$ & $a^\top x + \rho \| \nu \|_1 +  \tau \|B x - \nu \|_\infty \le b $ \\ \hline
		\end{tabular}
	\end{table}
	Note that since $\mathcal C$ is a polytope, and the uncertainty set $\mathcal D$ is a compact and convex 
	set with nonempty relative interior, then each semi-infinite constraint in~$\mathcal G$ admits a tractable robust counterpart, that is, it can be reformulated into a finite set of convex constraints by using standard techniques from robust optimization~\cite{ben09}; see Table~\ref{tab:g.examples} for examples. Note that if some of the linear equality constraints in~\eqref{eq:RobustDeePC} are certain, that is, not affected by uncertainties, then one can include such equalities in the set $\mathcal G$, in a similar way as in~\cite{huang2020quad}, instead of including them in the objective function of~\eqref{eq:minmaxDeePC-penalty}.

\begin{figure}
\begin{center}
\includegraphics[width=8.8cm]{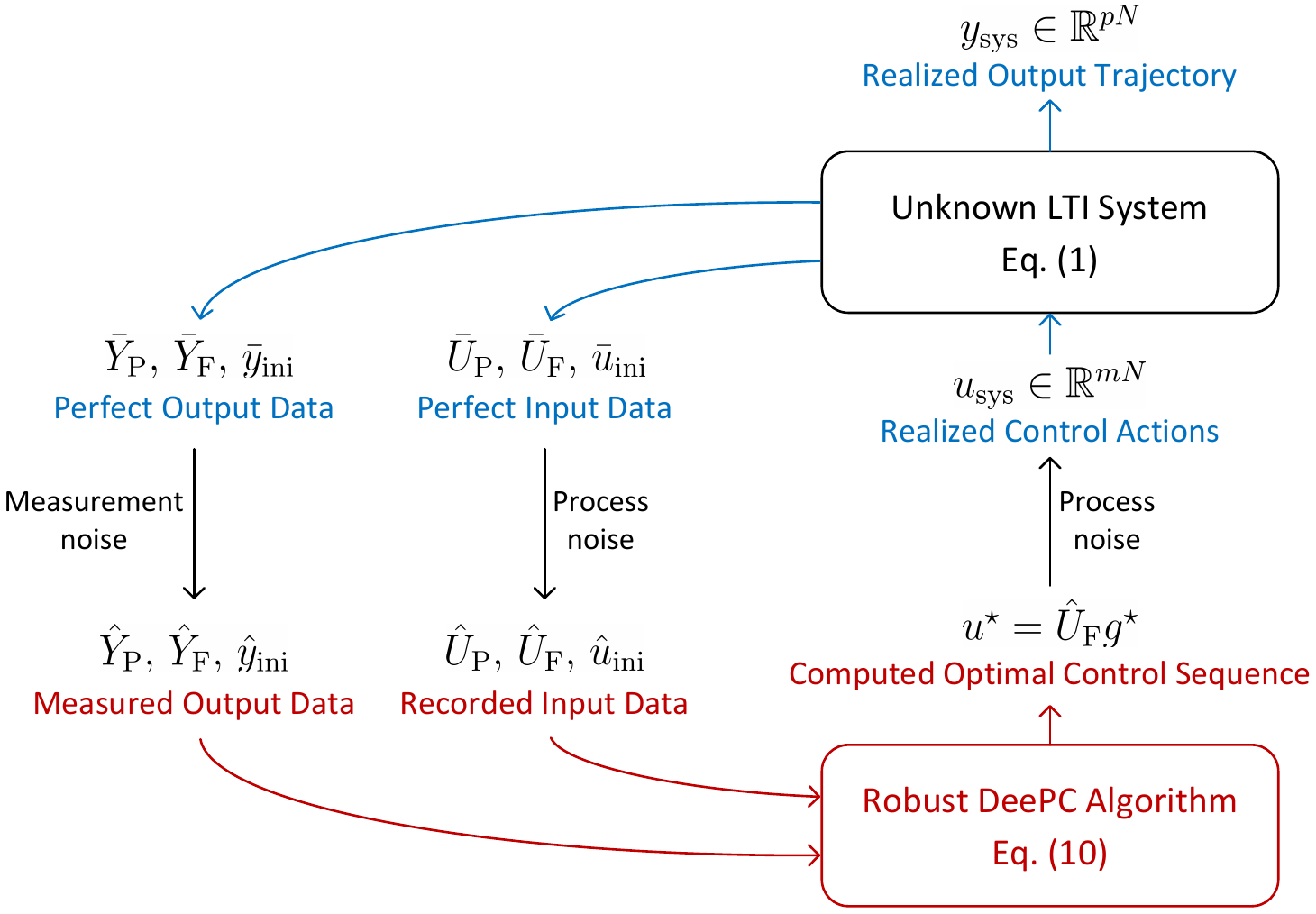}
\vspace{-5mm}
\caption{Interaction between the system and the robust DeePC algorithm. The data matrices $\hat U_{\rm P}$, $\hat Y_{\rm P}$, $\hat U_{\rm F}$, and $\hat Y_{\rm F}$ are obtained offline.}
\vspace{-4mm}
\label{Fig_data_loop}
\end{center}
\end{figure}

The minimizer $g^\star$ of~\eqref{eq:minmaxDeePC-penalty} can be used to compute the optimal control sequence $u^\star = \hat U_{\rm F} g ^\star$.
We illustrate the relationship between the unknown system and the robust DeePC algorithm in Fig.~\ref{Fig_data_loop}. Note that the realized control actions $u_{\rm sys}$ of the system may be different from $u^\star$ due to input disturbances. The realized output trajectory $y_{\rm sys}$ of the system (in response to $u_{\rm sys}$) is in general different from the predicted trajectory $y^\star = \hat Y_{\rm F}g^\star$ in the algorithm because the perfect data is not accessible to the controller. Note that $u^\star$ can also be applied in a receding horizon fashion as done for DeePC.

In the remainder of this paper, we focus on~\eqref{eq:minmaxDeePC-penalty} with a square quadratic estimation regularizer, that is,
	\begin{equation}\label{eq:Regl}
		h(\sigma_u, \sigma_y) = \lambda_{u} \| \sigma_u \|^2 + \lambda_{y} \| \sigma_y \|^2,
	\end{equation}
	where $\lambda_u$ and $\lambda_y$ are positive scalars.
	
	
\subsection{Performance Guarantees Induced by Robust DeePC}
	
	As a first result, we link the realized cost of the system, defined by $c_{\rm realized} = \left\| u_{\rm sys} \right\|_R^2 + \left\| y_{\rm sys} -r \right\|_Q^2$, to the optimization cost of \eqref{eq:minmaxDeePC-penalty}, and show how the realized cost can be certified by applying the robust DeePC. Here ``realized cost'' refers to the cost accrued by the real system if the optimizer of~\eqref{eq:minmaxDeePC-penalty} is applied in open loop for the whole horizon $N$ (see Fig.~\ref{Fig_data_loop}).
	We assume that the uncertainty set $\mathcal D$ contains the realization of the uncertain input/output data deviations.
	 \begin{assumption} \label{assum:D}
    There exists a $\bar{\xi} \in \mathcal D$ such that ${{U_{\rm P} (\bar{\xi})}} = \bar U_{\rm P}$, ${{Y_{\rm P} (\bar{\xi})}} = \bar Y_{\rm P}$, ${{U_{\rm F} (\bar{\xi})}} =\bar U_{\rm F}$, ${{Y_{\rm F} (\bar{\xi})}} = \bar Y_{\rm F}$, ${{u_{\rm ini} (\bar{\xi})}} = \bar u_{\rm ini}$ and ${{y_{\rm ini} (\bar{\xi}) }} = \bar y_{\rm ini}$.
    \end{assumption}

    Assumption~\ref{assum:D} ensures that the perfect input/output data is captured by the considered uncertainty set $\mathcal D$. However, even with Assumption~\ref{assum:D} in place, it is still not obvious that the realized cost can be guaranteed by solving \eqref{eq:minmaxDeePC-penalty} and then applying $u^\star = \hat U_{\rm F}g^\star$ to the system. This is because the violations of the initial conditions are penalized in \eqref{eq:minmaxDeePC-penalty}, and thus we generally have $\bar U_{\rm P}g \ne \bar u_{\rm ini}$ and $\bar Y_{\rm P}g \ne \bar y_{\rm ini}$, hence, the decision variable $g$ cannot be used to accurately predict the realized future trajectory. In the following theorem, we relate the optimal cost of the proposed robust DeePC algorithm~\eqref{eq:minmaxDeePC-penalty} to the realized cost on the unknown system.



		
\begin{theorem}\label{thm:real_cost}
If Assumption~\ref{assum:D} holds and $\|u_{\rm sys}  - \bar U_{\rm F}g^\star \| \le \eta_p$ where $g^\star$ is the minimizer of \eqref{eq:minmaxDeePC-penalty}, then there exists sufficiently large $(\lambda_u, \lambda_y) \in \mathbb R^2_{>0}$ such that
$$
2\sqrt{c_{\rm opt}} +  \eta_p \left( \sqrt{2} \left\| I_{mN}  \right\|_R  +  \left\|  {\mathcal{T}_N} \right\|_Q \right) \ge  \sqrt{c_{\rm realized}} \,,
$$
where $c_{\rm opt}$ denotes the minimum of \eqref{eq:minmaxDeePC-penalty}, and $\mathcal{T}_N$ is defined in \eqref{eq:T_N}.
\end{theorem}
		\begin{proof} 
		If Assumption~\ref{assum:D} holds, it follows that
			\begin{equation}\label{eq:real_perf1}
			\begin{split}
				c_{\rm opt} \ge & \left\|\bar U_{\rm F}g^\star \right\|_R^2 + \left\| \bar Y_{\rm F}g^\star-r \right\|_Q^2 + \lambda_u \left\| \bar U_{\rm P}g^\star - \bar u_{\rm ini} \right\|^2 \\
				&+ \lambda_y \left\| \bar Y_{\rm P}g^\star - \bar y_{\rm ini} \right\|^2  \,.
			\end{split}
			\end{equation}
		{Since the realized $(u_{\rm sys},y_{\rm sys})$ is an input/output trajectory departing from $(\bar u_{\rm ini}, \bar y_{\rm ini})$, according} to the Fundamental Lemma \cite{willems2005note}, there exists a {$\bar g$} that satisfies
			\begin{equation}\label{eq:true_system}
				\left[ {\begin{array}{*{20}{c}}
						{\bar U_{\rm P}}\\
						{\bar Y_{\rm P}}\\
						{\bar U_{\rm F}}\\
						{\bar Y_{\rm F}}
				\end{array}} \right]\bar g = \left[ {\begin{array}{*{20}{c}}
						{\bar u_{\rm ini}}\\
						{\bar y_{\rm ini}}\\
						u_{\rm sys}\\
						y_{\rm sys}
				\end{array}} \right].  
			\end{equation}
			By defining $\Delta_g = \bar g - g^\star$, we have
			\begin{equation}\label{eq:real_system1}
				\left[ {\begin{array}{*{20}{c}}
						{\bar U_{\rm P}}\\
						{\bar Y_{\rm P}}\\
						{\bar U_{\rm F}}\\
						{\bar Y_{\rm F}}
				\end{array}} \right] \Delta_g   = \left[ {\begin{array}{*{20}{c}}
						{\bar u_{\rm ini} - \bar U_{\rm P}g^\star}\\
						{\bar y_{\rm ini} - \bar Y_{\rm P}g^\star}\\
						u_{\rm sys} - \bar U_{\rm F}g^\star \\
						y_{\rm sys} - \bar Y_{\rm F}g^\star
				\end{array}} \right] =: \left[ {\begin{array}{*{20}{c}}
						\epsilon_{u_{\rm ini}}\\
						\epsilon_{y_{\rm ini}}\\
						\epsilon_{u_{\rm sys}} \\
						y_{\rm sys} - \bar Y_{\rm F}g^\star
				\end{array}} \right] \,,
			\end{equation}
			which implies that
			\begin{equation*}
			\begin{split}
				\Delta_g &= \begin{bmatrix}
					\bar U_{\rm P} \\
					\bar Y_{\rm P} \\
					\bar U_{\rm F}
				\end{bmatrix}^+\begin{bmatrix}
					\epsilon_{u_{\rm ini}}\\
					\epsilon_{y_{\rm ini}}\\
					\epsilon_{u_{\rm sys}}
				\end{bmatrix} + \begin{bmatrix}
					\bar U_{\rm P} \\
					\bar Y_{\rm P} \\
					\bar U_{\rm F}
				\end{bmatrix}^\bot x\,, \\
				\text{and} \quad \bar Y_{\rm F} \Delta_g &= \bar Y_{\rm F}\begin{bmatrix}
					\bar U_{\rm P} \\
					\bar Y_{\rm P} \\
					\bar U_{\rm F}
				\end{bmatrix}^+\begin{bmatrix}
					\epsilon_{u_{ \rm ini}}\\
					\epsilon_{y_{ \rm ini}}\\
					\epsilon_{u_{ \rm sys}}
				\end{bmatrix} =: K\epsilon_{\rm ini} + {\mathcal{T}_N} \epsilon_{u_{\rm sys}}  \,,
			\end{split}
			\end{equation*}
			where $x \in \mathbb{R}^{H_c}$ and $\epsilon_{\rm ini} = {\rm col}(\epsilon_{u_{\rm ini}}, \epsilon_{y_{\rm ini}})$ , while the second equality follows from $\bar Y_{\rm F}[{\rm col}(
				\bar U_{\rm P},
				\bar Y_{\rm P},
				\bar U_{\rm F})]
			^\bot = 0$ because the output  trajectory  is uniquely determined  given an initial trajectory and a future input trajectory \cite{markovsky2008data}. Note that $\begin{bmatrix} K & \mathcal{T}_N	\end{bmatrix}$ is the $N$-step auto-regressive matrix with extra input (ARX) of the system \eqref{eq:ABCD}. Since $\mathcal{T}_N$ maps the $N$-step future inputs to the $N$-step future outputs, it can be obtained as \eqref{eq:T_N}.

			By definding $\Lambda := {\rm diag}(\lambda_u I_{mT_{\rm ini}},\lambda_y I_{pT_{\rm ini}})$, it follows from~\eqref{eq:real_perf1} and~\eqref{eq:real_system1} that
			\begin{equation}\label{eq:real_perf2}
				\begin{split}
					c_{\rm opt}  \ge & \left\| u_{\rm sys} - \epsilon_{u_{\rm sys}} \right\|_R^2 \\
					& + \left\| y_{\rm sys} - K\epsilon_{  \rm ini} - { \mathcal{T}_N} \epsilon_{u_{  \rm sys}}  -r \right\|_Q^2 + \left\| \epsilon_{  \rm ini} \right\|^2_\Lambda \\
					\ge & \left\| u_{\rm sys} - \epsilon_{u_{ \rm sys}} \right\|_R^2 \\
					&+ \frac{1}{2}\left(\left\| y_{\rm sys} - r - K\epsilon_{  \rm ini} - {  \mathcal{T}_N} \epsilon_{u_{  \rm sys}}  \right\|_Q + \left\| \epsilon_{  \rm ini} \right\|_\Lambda  \right)^2\\
                   \ge & \left\| u_{\rm sys} - \epsilon_{u_{  \rm sys}} \right\|_R^2 \\
                   &+ \frac{1}{2}\left(  \left\| y_{\rm sys} - r  - {  \mathcal{T}_N} \epsilon_{u_{  \rm sys}} \right\|_Q  -  \left\| K\epsilon_{  \rm ini} \right\|_Q + \left\| \epsilon_{  \rm ini} \right\|_\Lambda  \right)^2 \\
                    \ge & \left\| u_{\rm sys} - \epsilon_{u_{  \rm sys}} \right\|_R^2 + \frac{1}{2} \left\| y_{\rm sys} - r - {  \mathcal{T}_N} \epsilon_{u_{  \rm sys}}  \right\|_Q^2 \,.
				\end{split}
			\end{equation}
			Here, the second inequality follows from the simple fact that~$a^2 + b^2 \ge \frac{1}{2} (a + b)^2$, the third inequality holds thanks to the reverse triangle inequality. The last inequality is satisfied if we take $\lambda_u$ and $\lambda_y$ large enough (by assumption) to ensure
            $$ \Lambda \succeq K^\top Q K  \quad \implies \quad \left\| \epsilon_{\rm ini} \right\|_\Lambda  \ge  \left\| K\epsilon_{\rm ini}    \right\|_Q .$$
       From~\eqref{eq:real_perf2}, we have
       \begin{equation*}
       c_{\rm opt} \ge \frac{1}{4} \left( \sqrt{2} \left\| u_{\rm sys} - \epsilon_{u_{ \rm sys}} \right\|_R +   \left\| y_{\rm sys} - r - { \mathcal{T}_N} \epsilon_{u_{ \rm sys}}  \right\|_Q  \right)^2 \,,
       \end{equation*}
       considering that ~$a^2 + b^2 \ge \frac{1}{2} (a + b)^2$. Further, we have
            \begin{equation*}
				\begin{split}
	   & 2 \sqrt{c_{\rm opt}}  \ge  \sqrt{2} \left\| u_{\rm sys} - \epsilon_{u_{ \rm sys}} \right\|_R +   \left\| y_{\rm sys} - r - { \mathcal{T}_N} \epsilon_{u_{ \rm sys}}  \right\|_Q  \\
	                & \ge  \sqrt{2} \left\| u_{\rm sys} \right\|_R - \sqrt{2} \left\| \epsilon_{u_{ \rm sys}} \right\|_R  +   \left\| y_{\rm sys} - r  \right\|_Q  -  \left\| { \mathcal{T}_N} \epsilon_{u_{ \rm sys}}  \right\|_Q \\
	                & \ge  \sqrt{2} \left\| u_{\rm sys} \right\|_R   +   \left\| y_{\rm sys} - r  \right\|_Q  - \eta_p \left( \sqrt{2} \left\| I_{mN}  \right\|_R  +  \left\| { \mathcal{T}_N} \right\|_Q \right)
				\end{split}
			\end{equation*}
		where the second inequality follows from the reverse triangle inequality, and the third due to the Cauchy–Schwarz inequality.
		Notice that $\sqrt{2} \left\| u_{\rm sys} \right\|_R   +   \left\| y_{\rm sys} - r  \right\|_Q \ge \sqrt{c_{\rm realized}}$ due to the simple fact that $a+b \ge \sqrt{a^2+b^2}$ when $a,b \ge 0$. This completes the proof.
		\end{proof}

	{\renewcommand\arraystretch{1.6}
	\begin{table*}[t]
		\centering
		\caption{Comparisons of different uncertainty sets for robust DeePC.}
		\label{tab:uncertainty_sets}
		\begin{tabular}{c|ccccc}
			\hline
			& \makecell{Uncertainty\\type}  & Tractability& \makecell{Setting different bounds\\for different columns} & \makecell{Removing the effects of\\ scaling matrices} & \makecell{Hankel\\ structure}  \rule{0pt}{4ex}\rule[-2.6ex]{0pt}{0pt} \\
			\hline
			Section~\ref{subsec:uns} & Unstructured  & Conic Quadratic  &  \XSolidBrush &  \XSolidBrush &  \XSolidBrush  \\
			\hline
			Section~\ref{subsec:col} & Column-wise  & Conic Quadratic & \Checkmark &  \XSolidBrush & \XSolidBrush  \\
			  \hline
			Section~\ref{subsec:int} & Interval  & Convex Quadratic & \Checkmark  & \Checkmark & \XSolidBrush  \\
			\hline
			Section~\ref{subsec:struct} & Structured  & Semi-definite  Program & \XSolidBrush  & \Checkmark & \Checkmark \\
			\hline
		\end{tabular}	
	\end{table*}
	}

{Theorem~\ref{thm:real_cost} shows that by solving the min-max problem in~\eqref{eq:minmaxDeePC-penalty}, one can obtain robust control input sequences that guarantee the realized input/output cost of the system, even with the initial conditions violated (i.e., we generally have $\bar U_{\rm P}g \ne \bar u_{\rm ini}$ and $\bar Y_{\rm P}g \ne \bar y_{\rm ini}$) in the min-max problem.} The choice of $(\lambda_u,\lambda_y)$ depends only on $Q$ and $K$, where $K$ is a constant matrix that is implicitly determined by the state-space matrices $A$, $B$, $C$, and $D$. Since Theorem~\ref{thm:real_cost} only requires $(\lambda_u,\lambda_y)$ to be sufficiently large, one can use a guess on the upper bound of the largest eigenvalue of $K^\top QK$ and choose $(\lambda_u,\lambda_y)$ accordingly. The value of $\eta_p$ increases with a higher level of input disturbance.
In the absence of input disturbance, that is, $u_{\rm sys} = \bar U_{\rm F}g^\star$, we construct a new upper bound on the realized input/output cost in the following result.

\begin{corollary}\label{thm:real_cost1}
If Assumption~\ref{assum:D} holds and the unknown system \eqref{eq:ABCD} does not have input disturbance,
then there exist sufficiently large $(\lambda_u, \lambda_y)  \in \mathbb R^2_{>0} $ such that
\begin{equation*}
 2c_{\rm opt} \ge c_{\rm realized}  \,.
\end{equation*}
\end{corollary}
\begin{proof}
Since $\epsilon_{u_{\rm sys}} = u_{\rm sys} - \bar U_{\rm F}g^\star = 0$ by assumption, the claim directly follows from \eqref{eq:real_perf2}.
\end{proof}
\noindent Note that with $\eta_p = 0$, the inequality in Corollary~\ref{thm:real_cost1} is at least as tight as the one from Theorem~\ref{thm:real_cost}.

In the absence of uncertainties in input/output data, one can consider $\mathcal D = \{0\}$, and thus \eqref{eq:minmaxDeePC-penalty} reduces to
\begin{equation*}
\begin{array}{cl}
\mathop{\min}\limits_{g}&\|\bar U_{\rm F}g\|_R^2 + \|\bar Y_{\rm F}g-r\|_Q^2\\
&+ \lambda_u \| \bar U_{\rm P}g - \bar u_{\rm ini}\|^2 + \lambda_y \| \bar Y_{\rm P}g - \bar y_{\rm ini}\|^2 \\
{\rm s.t.}& (\bar U_{\rm F}g,\bar Y_{\rm F}g) \in \mathcal{C} \,,
\end{array}
\end{equation*}
which recovers \eqref{eq:DeePC} by further considering $\bar U_{\rm P}g = \bar u_{\rm ini}$ and $\bar Y_{\rm P}g = \bar y_{\rm ini}$ as hard constraints. In this case and in the absence of disturbances, we have $c_{\rm opt} = c_{\rm realized}$.

The minimum of \eqref{eq:minmaxDeePC-penalty} will in general be increasing when a larger uncertainty set is needed to cover possible realized uncertainties, which, according to Theorem~\ref{thm:real_cost}, indicates that the (worst-case) realized input/output cost may also increase, as $c_{\rm opt}$ is always greater than $c_{\rm realized}$. In other words, the (worst-case) performance of robust DeePC degrades with increasing of noise level, and the best performance is achieved when perfect data is accessible. Moreover, one should also make the uncertainty set tight to reduce the conservativeness.

\section{Tractable Reformulations for Robust DeePC}
\label{sec:tra}
	In what follows, we derive tractable reformulations for \eqref{eq:minmaxDeePC-penalty} when various uncertainty sets are considered, and we will compare the conservativeness and the resulting robust counterparts with the different uncertainty sets. 
	Note that~\eqref{eq:minmaxDeePC-penalty}  can be compactly and equivalently represented as
	\begin{equation}
	\displaystyle \mathop {{\rm{min}}}\limits_{g \in \mathcal G} \max_{ \xi \in \mathcal D} \ \ \left\| A(\xi)g  -  b(\xi)  \right\|^2 \,,
	\label{eq:minmaxDeePC-penalty-rep1}
	\end{equation}
	where the elements of $A(\xi)$ and $b(\xi)$ are affine functions of $\xi$, namely,
	\[
	\displaystyle A (\xi) =  A^{(0)}  + \sum_{j= 1}^{n_\xi} A^{(j)} \xi_j,  \ {\rm and} \,\,\,  b (\xi ) = b_0 + \sum_{j= 1}^{n_\xi} b_j \xi_j\,,
	\]
	where
	\begin{equation*}
	\begin{split}
	A^{(j)} &= {\rm col}\left(\lambda_u^{\frac{1}{2}} U_{\rm P}^{(j)}  , \lambda_y^{\frac{1}{2}} Y_{\rm P}^{(j)} , R^{\frac{1}{2}}U_{\rm F}^{(j)} , Q^{\frac{1}{2}}Y_{\rm F}^{(j)}\right) \in \mathbb R^{H_r \times H_c} , \\
	b^{(j)} &= {\rm col} \left(\lambda_u^{\frac{1}{2}} u_{\rm ini}^{(j)} , \lambda_y^{\frac{1}{2}} y_{\rm ini}^{(j)} , 0 ,  0 \right) \in \mathbb R^{H_r} , \ \forall j \in [n_\xi] \,, \\
	A^{(0)} &= {\rm col}\left(\lambda_u^{\frac{1}{2}} \hat U_{\rm P}  , \lambda_y^{\frac{1}{2}} \hat Y_{\rm P} , R^{\frac{1}{2}}\hat U_{\rm F} , Q^{\frac{1}{2}}\hat Y_{\rm F}\right) \in \mathbb R^{H_r \times H_c} ,  \\
	b^{(0)} &= {\rm col} \left(\lambda_u^{\frac{1}{2}} \hat u_{\rm ini} , \lambda_y^{\frac{1}{2}} \hat y_{\rm ini}, 0 , Q^{\frac{1}{2}}r \right) \in \mathbb R^{H_r} ,
	\end{split}
	\end{equation*}
	and $H_r = (m+p)(T_{\rm ini} + N)$.
	
	Below we discuss various uncertainty quantifications and tractable reformulations sorted from coarse to fine, as listed in Table~\ref{tab:uncertainty_sets}. We will discuss the tractability and capabilities of the robust DeePC formulation \eqref{eq:minmaxDeePC-penalty-rep1} when these uncertainty sets are considered.


	\subsection{Unstructured Uncertainties}
	\label{subsec:uns}
	Consider the following variant of~\eqref{eq:minmaxDeePC-penalty-rep1} with unstructured uncertainties in the following form
	\begin{equation} \label{eq:unstructure}
		\displaystyle \mathop {{\rm{min}}}\limits_{g \in \mathcal G} \max_{ ( \Delta_A, \Delta_b) \in \mathcal D_{\rm uns}} \ \ \left\| ( A^{(0)} + \Delta_A) g  - ( b^{(0)} + \Delta_b) \right\|,
	\end{equation}
	where the unstructured uncertainty set is defined through
	\[
	\mathcal D_{\rm uns} = \{ ( \Delta_A, \Delta_b) \mid  \| [\Delta_{A} \ \Delta_b ]\|_F \le \rho_u \},
	\]
	with $\rho_u \in \mathbb R_{>0}$. Here we take the square-root of the quadratic objective function of~\eqref{eq:minmaxDeePC-penalty-rep1}, which will not affect the minimizer of the problem.
	
	Note that \eqref{eq:unstructure} is a special (albeit unstructured) case of~\eqref{eq:minmaxDeePC-penalty-rep1} with $\xi = {\rm col}\left( (\Delta_A)_1,(\Delta_A)_2,\dots,(\Delta_A)_{H_c},\Delta_b \right) \in \mathbb{R}^{H_r(H_c+1)}$ where $(\Delta_A)_i$ is the $i$-th column of $\Delta_A$, and it can be easily verified that the elements of $A(\xi) = A^{(0)}+\Delta_A$ and $b(\xi) = b^{(0)}+\Delta_b$ are affine functions of $\xi$. Moreover, Assumption~\ref{assum:D} can be easily satisfied by choosing a sufficiently large $\rho_u$, which leads to a performance certificate according to Theorem~\ref{thm:real_cost}.
	The following result shows that \eqref{eq:unstructure} can be reformulated into a tractable conic quadratic problem.
	
	\begin{proposition} \label{eq:frob}
		A vector $g^\star$ is a minimizer of~\eqref{eq:unstructure} if and only if $g^\star$ also minimizes the following conic quadratic problem
		\begin{equation}\label{eq:unstructure_1}
			\begin{array}{l}
				\mathop {{\rm{min}}}\limits_{g \in \mathcal{G}} \;\;\|A^{(0)} g - b^{(0)} \| +  \rho_u \sqrt{\|g\|^2 + 1}.
			\end{array}
		\end{equation}
		Moreover, the minima of~\eqref{eq:unstructure} and~\eqref{eq:unstructure_1} coincide.
	\end{proposition}
	\begin{proof}
		The proof is adapted from~\cite[Theorem~3.1]{el1997robust}. Consider a fixed $g$ in~\eqref{eq:unstructure}. It follows from triangle inequality that
		\begin{align*}
		   & \max_{ ( \Delta_A, \Delta_b) \in \mathcal D_{\rm uns}} \ \ \left\| ( A^{(0)} + \Delta_A) g  - ( b^{(0)} + \Delta_b) \right\| \\
	 \le  & \left\| A^{(0)}  g  - b^{(0)} \right\| + \max_{ ( \Delta_A, \Delta_b) \in \mathcal D_{\rm uns}} \left\|  \Delta_A g  - \Delta_b \right\| \\
	    = &  \left\| A^{(0)}  g  - b^{(0)} \right\| + \rho_u \sqrt{ \|g \|^2 + 1}.
		\end{align*}
	Choose $\hat{\Delta} = [\hat{\Delta}_A \ \hat{\Delta}_b] \in \mathcal D_{\rm uns}$ such that
	\[
	[\hat{\Delta}_A \ \hat{\Delta}_b]  = \frac{\rho_u \omega}{\sqrt{ \|g \|^2 + 1}} [ g^\top \ 1 ], \quad \text{where}
	\]
	\[
	\omega = \begin{cases}
	 \frac{A^{(0)}g - b^{(0)}}{{\|A^{(0)}g - b^{(0)}\|}} &  \text{if } A^{(0)}g \ne b^{(0)} \\
	 \text{any unit-norm vector} & \text{otherwise}.
	\end{cases}
	\]
	Since $\hat \Delta$ is a rank-one matrix, we have $\|\hat{\Delta} \|_F = \|\hat{\Delta} \| = \rho_u $, and
	\begin{align*}
	        & \left\| ( A^{(0)} + \hat{\Delta}_A) g  - ( b^{(0)} + \hat{\Delta}_b) \right\| \\
	    = & \left\| A^{(0)}  g  - b^{(0)} \right\| +  \left\| \hat{\Delta}_A g  - \hat{\Delta}_b \right\| \\
	    = & \left\| A^{(0)}  g  - b^{(0)} \right\| + \rho_u \sqrt{ \|g \|^2 + 1} \\
		\le & \max_{ ( \Delta_A, \Delta_b) \in \mathcal D_{\rm uns}} \ \ \left\| ( A^{(0)} + \Delta_A) g  - ( b^{(0)} + \Delta_b) \right\|,
	\end{align*}
	where the first equality follows from triangle inequality and the fact that $(A^{(0)}  g  - b^{(0)})$ is a nonnegative scalar of $(\hat{\Delta}_A g  - \hat{\Delta}_b)$. Since the lower bound coincides with the upper bound on the inner maximization problem~\eqref{eq:unstructure} for any fixed $g$, the claim follows.
	\end{proof}
	
	Similarly, if the vector $b$ in~\eqref{eq:unstructure}  is not affected by uncertainty, that is, $\Delta_b = 0$, then the corresponding tractable reformulation becomes \cite[Theorem~2]{bertsimas2018characterization}
	\begin{equation*}
		\begin{array}{l}
			\mathop {{\rm{min}}}\limits_{g \in \mathcal{G}} \;\;\|A^{(0)} g - b^{(0)} \| +  \rho_u \|g\|.
		\end{array}
	\end{equation*}
	The min-max formulation~\eqref{eq:unstructure} is equivalent to the regularization~\eqref{eq:unstructure_1}, which reminiscent of other DeePC regularizations studied in~\cite{dorfler2021bridging,berberich2019data,coulson2020distributionally,huang2020quad}. We will explore these connections in Section~\ref{sec:regl}.
	
	We remark that the considered uncertainties in \eqref{eq:unstructure} can be overly conservative for several reasons: 1) the constructed worst-case realization of the uncertainties satisfies rank$([\Delta_A\;\Delta_b]) = 1$ \cite{el1997robust}, which cannot be achieved in general when Hankel matrices are used as predictors; 2) Hankel structures cannot be imposed on $\Delta_A$; 3) the uncertainties are indirectly added to the input/output data through the scaling matrices, e.g., $Q^{\frac{1}{2}}$ and $\lambda_y^{\frac{1}{2}}I$; 4) the columns in $[\Delta_A\;\Delta_b]$ are uncorrelated when Page matrices or trajectory matrices are used, but $\mathcal D_{\rm uns}$ considers the Frobenius norm of $[\Delta_A\;\Delta_b]$ to be bounded which may impose correlations among the columns. In the following sections we alleviate these sources of conservatism. First we consider column-wise uncertainties for $[\Delta_A\;\Delta_b]$.

	\subsection{Generalized Column-wise Uncertainties}
	\label{subsec:col}
	Consider the following variant of~\eqref{eq:minmaxDeePC-penalty-rep1} with generalized column-wise uncertainties, that is,
	\begin{equation} \label{eq:column}
		\displaystyle \mathop {{\rm{min}}}\limits_{g \in \mathcal G} \max_{(\Delta_A,\Delta_b) \in \mathcal D_{\rm gcol}} \ \ \left\| ( A^{(0)} + \Delta_A) g  -  (b^{(0)}+\Delta_b) \right\|.
	\end{equation}
	Here the generalized column-wise uncertainty set~\cite[Section 3]{xu2010robust} is defined column-by-column as
	\[
	\begin{split}
	\mathcal D_{\rm gcol}\hspace{-0.7mm} =\hspace{-0.75mm} \{&(\Delta_A,\Delta_b) \mid \exists \rho_A \in \mathbb R^{H_c}\hspace{-1mm} :\hspace{-0.5mm}  \| (\Delta_A)_{i} \| \le (\rho_{A})_i,  \forall i \in [H_c], \\
	&f_j ( \rho_A ) \le 0, \ \forall j \in [J], \Delta_b \le \rho_b \},
	\end{split}
	\]
	where the function $f_j:\mathbb R^{H_c} \rightarrow \mathbb [-\infty, +\infty]$ is a proper, closed and convex function for each $j \in [J]$, and $(\Delta_A)_{i}$ is the $i$-th column of the matrix $\Delta_A$.
	
	\begin{proposition} \label{eq:column-wise}
		If $\mathcal D_{\rm gcol}$ is nonempty and admits a Slater point, then a vector $g^\star$ is a minimizer of~\eqref{eq:column} if and only if there exists a $(\lambda^\star, \{y_j^\star \}_{j=1}^J)$ such that $g^\star$ also minimizes the following convex optimization problem
		\begin{equation}\label{eq:colwise}
		\begin{array}{cll}
		   \displaystyle   \min_{g \in \mathcal{G}, \lambda  \atop  \{y_j\}_{j=1}^J} & \displaystyle \|A^{(0)} g-b^{(0)} \| + \sum_{j \in [J]} \psi_j ( y_j, \lambda_j) + \rho_b\\
		     {\rm s.t.}  &\lambda_j \ge 0 \quad \forall j \in [J] \\
		     &  \displaystyle \sum_{j \in [J]} y_j \le - |g| ,
		\end{array}
		\end{equation}
		where $\psi: \mathbb R^{n_{x}} \times \mathbb R_{\ge 0} \rightarrow  [-\infty, +\infty]$ denotes the perspective function\footnote{We adopt the definition of the perspective function from \cite{Rockafellar1970}, where $\psi(x, t) = t f^*(x/t)$ if $t>0$, and $\psi(x, t) = \delta^*_{{\rm dom} (f)} (x)$ if $t = 0$, and $\delta^*_{\mathcal S}$ denotes the support function over the set $\mathcal S$. The perspective function $\psi$ is proper, convex and closed if and only if $f^*$ is proper, convex and closed.} of the conjugate function $f^*:\mathbb R^{n_{x}} \rightarrow  [-\infty, +\infty]$ of~$f$. Moreover, the minima of~\eqref{eq:column} and~\eqref{eq:colwise} coincide.
	\end{proposition}
	
	\begin{proof}
	The proof is adapted from \cite[Theorem 1]{xu2010robust}. Fix $g$ in~\eqref{eq:column}. We first prove that
	\begin{align*}
	  & \max_{(\Delta_A,\Delta_b) \in \mathcal D_{\rm gcol}} \ \ \left\| ( A^{(0)} + \Delta_A) g  -  (b^{(0)}+\Delta_b) \right\| \\
	= & \|A^{(0)} g - b^{(0)} \| +  \displaystyle \max_{\rho_A \in \mathcal V} \  \rho_A^\top |g| + \rho_b,
	\end{align*}
	where $\mathcal V = \{\rho_A \in \mathbb R^{H_c}_{>0} \mid \ f_j ( \rho_A ) \le 0, \ \forall j \in [J] \}$.
    Note that
    \begin{align*}
	  & \max_{(\Delta_A,\Delta_b) \in \mathcal D_{\rm gcol}} \ \ \| ( A^{(0)} + \Delta_A) g  -  (b^{(0)}+\Delta_b) \| \\
	= & \max_{(\Delta_A,\Delta_b) \in \mathcal D_{\rm gcol}} \ \ \| A^{(0)} g  -  b^{(0)} + \sum_{i \in [H_r]} (\Delta_A)_i g_i - \Delta_b\| \\
	\le  & \| A^{(0)} g  -  b^{(0)}\| + \max_{(\Delta_A,\Delta_b) \in \mathcal D_{\rm gcol}}  \sum_{i \in [H_r]} \|  (\Delta_A)_i g_i \| + \|\Delta_b\| \\
	\le  &  \| A^{(0)} g  -  b^{(0)}\| + \max_{\rho_A \in \mathcal V}  \ \rho_A^\top | g | + \rho_b.
	\end{align*}
	To prove the converse inequality, choose a $\hat \Delta = [\hat{\Delta}_A \ \hat{\Delta}_b] \in \mathcal D_{\rm gcol}$ such that
	\[
	(\hat{\Delta}_A)_i  = - (\rho_A)_i {\rm sgn}(g_i) \omega \  {\rm and} \ \hat{\Delta}_b  = \rho_b \omega, \quad \text{where}
	\]
	\[
	\omega = \begin{cases}
	 \frac{A^{(0)}g - b^{(0)}}{{\|A^{(0)}g - b^{(0)}\|}} &  \text{if} A^{(0)}g \ne b^{(0)} \\
	 \text{any unit-norm vector} & \text{otherwise}.
	\end{cases}
	\]
	Then, we have
    \begin{align*}
	  & \max_{(\Delta_A,\Delta_b) \in \mathcal D_{\rm gcol}} \ \ \| (b^{(0)}+\Delta_b) - ( A^{(0)} + \Delta_A) g    \| \\
	\ge &  \max_{\rho_A \in \mathcal V}  \| b^{(0)} + \hat{\Delta}_b -   A^{(0)} g - \sum_{i \in [H_r]} (\hat{\Delta}_A)_i g_i \| \\
	=  &  \max_{\rho_A \in \mathcal V} \|  b^{(0)} + \rho_b \omega - A^{(0)} g  + \sum_{i \in [H_r]}  g_i  (\rho_A)_i {\rm sgn}(g_i) \omega \| \\
	=  &  \max_{\rho_A \in \mathcal V} \|  b^{(0)} + \rho_b \omega - A^{(0)} g  + \sum_{i \in [H_r]}    \rho_A^\top |g|  \omega \| \\
	= & \|A^{(0)} g - b^{(0)} \| +  \displaystyle \max_{\rho_A \in \mathcal V} \  \rho_A^\top |g| + \rho_b.
	\end{align*}
	Finally, one can reformulate the inner maximization into its dual and obtain~\eqref{eq:colwise} thanks to the strong duality of convex programs, which applies because $\mathcal D_{\rm gcol}$ is nonempty and admits a Slater point.
	\end{proof}
	
	We remark that if the column-wise uncertainty is assumed, one can conservatively approximate~\eqref{eq:column} via~\eqref{eq:unstructure} with a large enough $\mathcal D_{\rm uns}$, e.g., if $\rho_u^2 \ge \|\rho_A\|^2+\rho_b^2$ in $\mathcal D_{\rm uns}$, which implies that the minimum of \eqref{eq:column} is smaller or equal than that of \eqref{eq:unstructure}.
	Problem~\eqref{eq:column} is useful if unstructured Page matrices or trajectory matrices are considered as predictors. Furthermore, different bounds for different columns in the data matrices can be incorporated in~\eqref{eq:column}, which comes in handy when the data comes from a time-varying system or the columns are from independent experiments (i.e., trajectory matrices are used). We relegate this direction to future research.

	
	The following result shows that for a special case of~\eqref{eq:column}, where $\rho_A \in \mathbb{R}^{H_c}$ is a constant vector in $\mathcal D_{\rm gcol}$, then the objective function of~\eqref{eq:column} can be reformulated into a least-square problem with a 1-norm regularization.
	
	\begin{corollary}\label{co:one_norm}
	Consider $\mathcal D_{\rm gcol} = \mathcal D_{\rm col}$ in~\eqref{eq:column}, where
	\begin{equation}\label{eq:column_wise_set}
	\mathcal D_{\rm col} := \{(\Delta_A,\Delta_b) \mid   \| (\Delta_A)_{i} \| \le (\rho_{A})_i,  \forall i \in [H_c],
	\Delta_b \le \rho_b \},
	\end{equation}
	then $g^\star$ is a minimizer of~\eqref{eq:column} if and only if $g^\star$ minimizes
	\begin{equation}\label{eq:colwise2}
			\mathop {{\rm{min}}}\limits_{g \in \mathcal{G}} \;\;\|A^{(0)} g - b^{(0)} \| + \rho_A^\top |g| + \rho_b\,.
	\end{equation}
	Moreover, the minima of~\eqref{eq:column} and~\eqref{eq:colwise2} coincide.
	\end{corollary}
	\begin{proof}
	The claimed result follows from Proposition~\ref{eq:column-wise}.
	\end{proof}
	

		
	Since $(A^{(0)})_i$ is a weighted input/output trajectory (with scaling matrices $\lambda_y^{\frac{1}{2}}I$, $Q^{\frac{1}{2}}$, etc.), it could be conservative to consider a column-wise uncertainty vector $(\Delta_A)_i$ (with 2-norm bound) on $(A^{(0)})_i$ as the uncertainties are indirectly added to the input/output data through the scaling matrices. We will next demonstrate how to use interval uncertainties to incorporate the effects of the scaling matrices.
	
	\subsection{Interval Uncertainties}
	\label{subsec:int}
	Consider the following special case of~\eqref{eq:minmaxDeePC-penalty-rep1}
	\begin{equation} \label{eq:interval}
		\displaystyle \mathop {{\rm{min}}}\limits_{g \in \mathcal G} \max_{  (\Delta_A, \Delta_b) \in \mathcal D_{\rm int}} \ \ \left\| ( A^{(0)} + \Delta_A) g  - ( b^{(0)} + \Delta_b) \right\|
	\end{equation}
	with the nonempty interval uncertainty set
	\[
	\mathcal D_{\rm int} = \{(\Delta_A, \Delta_b) \mid  | (\Delta_A)_{ij} | \le \bar A_{ij},  \forall (i,j),  | (\Delta_b)_i | \le   \bar b_i ,   \forall i\} ,
	\]
	where $\bar A \in \mathbb R^{H_r \times H_c}_{\ge 0}$, $\bar b \in \mathbb R^{H_r}_{\ge 0}$, and $(\Delta_A)_{ij}$ is the element in the $i$-th row and $j$-th column of~$\Delta_A$. The element of $\bar A$ and~$\bar b$ can be set to $0$ if there is no uncertainty in the corresponding element of the data matrices. The corresponding $\mathcal{G}$ in~\eqref{eq:interval} can be represented as a polyhedron.
	Note that if the interval uncertainty is assumed, one can still (conservatively) capture it with a large enough $\mathcal D_{\rm col}$, e.g., $(\rho_A)_i \ge \|(\bar A)_i\| \ \forall i \in [H_c]$ and $\rho_b \ge \|\bar b\|$ in $\mathcal D_{\rm col}$. In this case, the minimum of \eqref{eq:interval} is less or equal than that of \eqref{eq:column}, suggesting a less conservative result (better performance in the worst-case scenario) by applying $\mathcal D_{\rm int}$. The following result shows that \eqref{eq:interval} can be reformulated as a convex quadratic problem.
	
	\begin{proposition} \label{eq:intval}
		A vector $g^\star$ is a minimizer of~\eqref{eq:interval} if and only if there exists a $(\gamma^\star, \nu^\star)$ such that $(g^\star, \gamma^\star, \nu^\star)$ also minimizes the following convex quadratic problem
			\begin{equation}\label{eq:interval_QP}
			\begin{array}{cl}
				\displaystyle \min_{g \in \mathcal{G}, \gamma, \nu } & \;\; \| \gamma + \bar b +  \bar A \nu \|^2  \\
				{\rm s.t. } &  \;\; \displaystyle -\gamma \le  A^{(0)} g - b^{(0)} \le \gamma  \\
			    &  \;\; 	-\nu \le g \le \nu .
			\end{array}
		\end{equation}
		Moreover, the minimum of~\eqref{eq:interval} coincides with $\| \gamma^\star + \bar b + \bar A \nu^\star \|$.
	\end{proposition}
	\begin{proof}
	This proof is adapted from \cite[Section 6.2]{ben09}. One may observe that~\eqref{eq:interval} can be equivalently reformulated as
		\begin{equation*}
			\begin{array}{cl}
				\displaystyle \min_{g \in \mathcal{G}, \tau } &  \tau^\top \tau \\
				{\rm s.t. } &   \displaystyle   | (A^{(0)}g - b^{(0)})_i | + \bar b_i + \sum_{j = 1}^{H_c} \bar A_{ij} |g_j| \le \tau_i, \ \forall i \in [H_r],
			\end{array}
		\end{equation*}
		where $(x)_i$ denotes the $i$-th element of the vector $x$,
		and the minimum of~\eqref{eq:interval} coincides with $\| \tau^\star\|$. One can now replace $\tau$ by the introduced auxiliary variables $\gamma_i$ and $\nu_j$ for $| (A^{(0)}g - b^{(0)})_i|$ and $|g_j|$, respectively, for every $i \in [H_r]$ and $j \in [H_c]$.
	\end{proof}	

		The interval uncertainty set allows us to consider independent upper and lower bounds for the possible deviation in each entry of $A^{(0)}$ and $b^{(0)}$. Hence, one can directly consider uncertainties on the input/output data and cancel the effects of the scaling matrices.
		For example, one may consider the measurement error of the $i{\rm th}$ output to be bounded by $\tilde y_i \in \mathbb{R}_{\ge 0}$, and the process noise of the $i{\rm th}$ input to be bounded by $\tilde u_i \in \mathbb{R}_{\ge 0}$ (which can also be zero if no uncertainty occurs). In this case, we have
		\begin{equation}\label{eq:interval_bound}
			\bar A = \begin{bmatrix}
				\lambda_u^{\frac{1}{2}}{\bf 1}_{T_{\rm ini}\times H_c} \otimes \tilde u \\
				\lambda_y^{\frac{1}{2}}{\bf 1}_{T_{\rm ini}\times H_c} \otimes \tilde y \\
				R^{\frac{1}{2}} {\bf 1}_{N\times H_c} \otimes \tilde u \\
				Q^{\frac{1}{2}} {\bf 1}_{N\times H_c} \otimes \tilde y
			\end{bmatrix} \quad {\rm and} \quad
			\bar b = \begin{bmatrix}
				\lambda_u^{\frac{1}{2}}{\bf 1}_{T_{\rm ini}} \otimes \tilde u \\
				\lambda_y^{\frac{1}{2}}{\bf 1}_{T_{\rm ini}} \otimes \tilde y \\
				0 \\
				0
			\end{bmatrix}  \,.
		\end{equation}
		In the above setting, the possible deviation in each entry of $A^{(0)}$ and $b^{(0)}$ is assumed to reside within an independent interval, and thus Hankel structures for uncertainties in $A(\xi)$ cannot be enforced. {For this reason,} in contrast to the case in which Page matrices or trajectory matrices are considered, if Hankel matrices are used in \eqref{eq:interval_QP} as predictors, then the obtained solution could be overly conservative, similar to the cases with unstructured uncertainties or column-wise uncertainties. To this end, in what follows we illustrate how the uncertainties with Hankel structures can be taken into account in the robust DeePC algorithm with a tractable formulation.

	\subsection{Structured Uncertainties}
	\label{subsec:struct}
	Consider~\eqref{eq:minmaxDeePC-penalty-rep1} with the uncertainties residing within the following uncertainty set
	\[{\mathcal D} = \mathcal D_{ \rm struct} := \{ \xi \in \mathbb R^{n_\xi} \ | \ \| \xi \| \le \rho_s \}.   \]
	In this case, the corresponding $\mathcal G$ is a convex set that consists of linear and conic quadratic constraints. Recall that the vector of uncertainties $\xi$ enters the elements of $A(\xi)$ and $b(\xi)$ with affine structures, and we will show below how these affine structures allow us to construct Hankel structures for the uncertainties.

	We start by noting that \eqref{eq:minmaxDeePC-penalty-rep1} can be reformulated as
	\begin{equation}
	\displaystyle \mathop {{\rm{min}}}\limits_{g \in \mathcal G} \max_{ \xi \in \mathcal D} \ \ \left\| D(g) \xi  - c(g)  \right\|^2, \label{eq:minmaxDeePC-penalty-rep2}
	\end{equation}
	where the elements of $D(g)$ and $c(g)$ that are affine functions of $g$, namely,
	\[
	D (g) =  D^{(0)}  + \sum_{\ell= 1}^{H_c}  D^{(\ell)}g_\ell,  \ \text{and} \ c (g) = c_0 + \sum_{\ell= 1}^{H_c} c_\ell g_\ell.
	\]
	The equivalence between~\eqref{eq:minmaxDeePC-penalty-rep1} and~\eqref{eq:minmaxDeePC-penalty-rep2} directly follows from
	\begin{align*}
		(A^{(0)})_\ell =& c_\ell \ \forall \ell \in [H_c],  \ (D^{(0)})_j = b_j \ \forall j \in [n_\xi], \ b_0 = c_0, \\
		\quad (A^{(j)})_\ell =& (D^{(\ell)})_j \ \forall \ell \in [H_c] \ \forall j \in [n_\xi],
	\end{align*}
	where, for instance, $(A^{(0)})_\ell$ is the $\ell$-th column of $A^{(0)}$.
	
	
	\begin{example}[Additive Uncertainties {with Hankel Structures in~\eqref{eq:minmaxDeePC-penalty-rep2}}]\label{Example_1}
		Consider uncertainties directly on the time series recorded input data and measured output data (i.e., $\hat u^{\rm d}$, $\hat u_{\rm ini}$, $\hat y^{\rm d}$, and $\hat y_{\rm ini}$). Let $\xi = {\rm col}(\xi_1 \in \mathbb{R}^{mT},\xi_2 \in \mathbb{R}^{pT},\xi_3 \in \mathbb{R}^{mT_{\rm ini}},\xi_4 \in \mathbb{R}^{pT_{\rm ini}}) \in \mathcal{D}_{ \rm struct}$ such that
		\begin{equation*}
			\begin{bmatrix} U_{\rm P}(\xi)\\ U_{\rm F}(\xi) \end{bmatrix} = \mathscr{H}_{T_{\rm ini}+N}(\hat u^{\rm d} + \alpha_1 \xi_1) = \begin{bmatrix} \hat U_{\rm P}\\ \hat U_{\rm F} \end{bmatrix} + \alpha_1 \mathscr{H}_{T_{\rm ini}+N}(\xi_1)\,,
		\end{equation*}
		\begin{equation*}
			\begin{bmatrix} Y_{\rm P}(\xi)\\ Y_{\rm F}(\xi) \end{bmatrix} = \mathscr{H}_{T_{\rm ini}+N}(\hat y^{\rm d} + \alpha_2 \xi_2) = \begin{bmatrix} \hat Y_{\rm P}\\ \hat Y_{\rm F} \end{bmatrix} + \alpha_2 \mathscr{H}_{T_{\rm ini}+N}(\xi_2)\,,
		\end{equation*}
		\begin{equation*}
			u_{\rm ini}(\xi) = \hat u_{\rm ini} + \alpha_3 \xi_3\,,
			y_{\rm ini}(\xi) = \hat y_{\rm ini} + \alpha_4 \xi_4\,,
		\end{equation*}
		where $\alpha_1$, $\alpha_2$, $\alpha_3$, and $\alpha_4$ are the scaling coefficients. {Note that $\alpha_1 = \alpha_3 = 0$ in the absence of input disturbance.}
		
		We now show that~\eqref{eq:minmaxDeePC-penalty-rep2} with additive uncertainties in the input/output data matrices (that obey the prescribed Hankel structures) admit a compact representation. Notice that
		\begin{equation*}
		\begin{split}
			\begin{bmatrix} U_{\rm P}(\xi)\\ U_{\rm F}(\xi) \end{bmatrix} g =&
			\begin{bmatrix} \hat U_{\rm P}\\ \hat U_{\rm F} \end{bmatrix} g + \alpha_1 \mathscr{H}_{T_{\rm ini}+N}(\xi_1)g \\
			=& \begin{bmatrix} \hat U_{\rm P}\\ \hat U_{\rm F} \end{bmatrix} g + \alpha_1 \mathscr{M}_{T_{\rm ini}+N}(g) \otimes I_m \xi_1\,,
		\end{split}
		\end{equation*}
		\begin{equation*}
		\begin{split}
			\begin{bmatrix} Y_{\rm P}(\xi)\\ Y_{\rm F}(\xi) \end{bmatrix} g =&
			\begin{bmatrix} \hat Y_{\rm P}\\ \hat Y_{\rm F} \end{bmatrix} g + \alpha_2 \mathscr{H}_{T_{\rm ini}+N}(\xi_2)g \\
			=& \begin{bmatrix} \hat Y_{\rm P}\\ \hat Y_{\rm F} \end{bmatrix} g + \alpha_2 \mathscr{M}_{T_{\rm ini}+N}(g) \otimes I_p \xi_2\,,
		\end{split}
		\end{equation*}
		where $\mathscr{M}_L(x) \in \mathbb{R}^{L \times (n_x+L-1)}$ ($x \in \mathbb{R}^{n_x}$) is defined as
		\begin{equation*}
			\mathscr{M}_L(x) = \begin{bmatrix}
				x_1 & x_2 & \cdots & x_L     & \cdots & x_{n_x} & & &   \\
				& x_1 & x_2    & \cdots  & x_L    & \cdots  & x_{n_x} & &   \\
				&     & \ddots & \ddots  &        & \ddots  &         & \ddots & \\
				&     &        & x_1     & x_2    & \cdots  & x_L     & \cdots & x_{n_x}
			\end{bmatrix} .
		\end{equation*}
		One can further partition $\mathscr{M}_{T_{\rm ini}+N}(g)$ into
		\begin{equation*}
			\begin{bmatrix} \mathscr{M}_{T_{\rm ini}+N}^{\rm P}(g) \in \mathbb{R}^{T_{\rm ini} \times T}\\[6pt] \mathscr{M}_{T_{\rm ini}+N}^{\rm F}(g) \in \mathbb{R}^{N \times T} \end{bmatrix} := \mathscr{M}_{T_{\rm ini}+N}(g)\,.
		\end{equation*}
		Since
		\begin{equation*}
			A(\xi) = \begin{bmatrix}\lambda_u^{\frac{1}{2}}U_{\rm P}(\xi) \\\lambda_y^{\frac{1}{2}}Y_{\rm P}(\xi) \\ R^{\frac{1}{2}}U_{\rm F}(\xi) \\ Q^{\frac{1}{2}}Y_{\rm F}(\xi) \end{bmatrix} \quad {\text and} \quad
			b(\xi) = \begin{bmatrix} \lambda_u^{\frac{1}{2}}u_{\rm ini}(\xi) \\\lambda_y^{\frac{1}{2}}y_{\rm ini}(\xi)\\ 0 \\ Q^{\frac{1}{2}}r \end{bmatrix}\,
		\end{equation*}
		by definition, we then have
		\begin{equation}\label{eq:sdp_Dc}
		\begin{split}
		D(g) = & \begin{bmatrix} D_1(g) & D_2(g) \end{bmatrix}\,, \\
		c(g) = & \begin{bmatrix}\lambda_u^{\frac{1}{2}}\hat U_{\rm P} \\\lambda_y^{\frac{1}{2}}\hat Y_{\rm P} \\ R^{\frac{1}{2}}\hat U_{\rm F} \\ Q^{\frac{1}{2}}\hat Y_{\rm F} \end{bmatrix}g - \begin{bmatrix} \lambda_u^{\frac{1}{2}}\hat u_{\rm ini} \\\lambda_y^{\frac{1}{2}}\hat y_{\rm ini}\\ 0 \\ Q^{\frac{1}{2}}r \end{bmatrix}\,,
		\end{split}
		\end{equation}
        where
		\begin{align*}
		&	D_1(g) \hspace{-0.5mm} = \hspace{-1mm} \begin{bmatrix}
				\lambda_u^{\frac{1}{2}}\alpha_1 \mathscr{M}^{\rm P}_{T_{\rm ini}+N}(g) \otimes I_m& 0 \\
				0&\lambda_y^{\frac{1}{2}}\alpha_2 \mathscr{M}^{\rm P}_{T_{\rm ini}+N}(g) \otimes I_p  \\
				R^{\frac{1}{2}}\alpha_1 \mathscr{M}^{\rm F}_{T_{\rm ini}+N}(g) \otimes I_m& 0 \\
				0&Q^{\frac{1}{2}}\alpha_2 \mathscr{M}^{\rm F}_{T_{\rm ini}+N}(g) \otimes I_p
			\end{bmatrix} \\
    	&	\text{and} \quad	D_2(g) = \begin{bmatrix}
				\lambda_u^{\frac{1}{2}}\alpha_3 I_{mT_{\rm ini}}  & 0 \\
				0& \lambda_y^{\frac{1}{2}}\alpha_4 I_{pT_{\rm ini}} \\
				0&  0\\
				0&0
			\end{bmatrix}\,.
		\end{align*}
	\end{example}
	\vspace{3mm}
	
	Example~\ref{Example_1} shows that the formulation in~\eqref{eq:minmaxDeePC-penalty-rep2} admits a compact representation of uncertainties with Hankel structures, as derived in~\eqref{eq:sdp_Dc}.
	The following result explicitly shows how~\eqref{eq:minmaxDeePC-penalty-rep2} with $\mathcal D = \mathcal D_{\rm struct}$ can be solved by considering a tractable semi-definite programming problem.
	
	\begin{proposition} \label{pro:SDP}
		A vector $g^\star$ is a minimizer of~\eqref{eq:minmaxDeePC-penalty-rep2} with $\mathcal D = \mathcal D_{ \rm struct}$ if and only if there exists a $(\tau^\star, \lambda^\star)$ such that $g^\star$ also minimizes the following semi-definite programming problem
		\begin{equation}\label{eq:sdp}
			\begin{array}{cll}
				\displaystyle \min_{g \in \mathcal G, \tau, \lambda } & \;\; \tau \\
				{\rm s.t. } &  \;\; \begin{bmatrix}
					\tau - \lambda \rho_s^2 & 0 & c(g)^\top \\
					0 & \lambda I & D(g)^\top \\
					c(g) &  D(g) & I
				\end{bmatrix} \succeq 0.
			\end{array}
		\end{equation}
		Moreover, the minima of~\eqref{eq:minmaxDeePC-penalty-rep2} and~\eqref{eq:sdp} coincide.
	\end{proposition}
	\begin{proof}
	 The proof is adapted from~\cite{el1997robust}. For any fixed $g$, we have
		\begin{align*}
		    & \tau \ge \max_{ \xi \in \mathcal D} \left\| D(g) \xi  - c(g)  \right\|^2 \\
     \iff \quad  & \tau \ge \left\| D(g) \xi  - c(g)  \right\|^2  \quad \forall \xi: \| \xi \| \le \rho_s.
		\end{align*}
	This semi-infinite constraint is satisfied if and only if
	\[
	\exists \lambda \ge 0: \begin{bmatrix}
					\tau - \lambda \rho_s^2 - c(g)^\top c(g) & -c(g)^\top D(g)  \\
					-D(g)^\top c(g) & \lambda I - D(g)^\top D(g)
				\end{bmatrix} \succeq 0,
	\]
	thanks to S-lemma \cite[\S 2.3]{Yakubovich:1971}, which applies because $\rho_s \in  \mathbb R_{>0}$. The inner maximization problem in~\eqref{eq:minmaxDeePC-penalty-rep2} can be reformulated as
		\begin{equation*}
			\begin{array}{cll}
				\displaystyle \min_{\tau, \lambda} & \;\;  \tau   \\
				{\rm s.t. } &  \;\; \begin{bmatrix}
					\tau - \lambda \rho_s^2 - c(g)^\top c(g) & -c(g)^\top D(g)  \\
					-D(g)^\top c(g) & \lambda I - D(g)^\top D(g)
				\end{bmatrix} \succeq 0.
			\end{array}
		\end{equation*}
		It then follows from Schur complement \cite[\S A.5.5]{Boyd:2004} that
		\begin{align*}
		    \begin{bmatrix}
					\tau - \lambda \rho_s^2 & 0  \\
					0 & \lambda I_{H_r}
			\end{bmatrix} - \begin{bmatrix}
			c(g)^\top \\ D(g)^\top
			\end{bmatrix}  I_{H_r}
			\begin{bmatrix}
			c(g) & D(g)
			\end{bmatrix} \succeq 0
		\end{align*}
		if and only if
		\[
		\begin{bmatrix}
					\tau - \lambda \rho_s^2 & 0 & c(g)^\top \\
					0 & \lambda I & D(g)^\top \\
					c(g) &  D(g) & I_{H_r}
		\end{bmatrix} \succeq 0,
		\]
		and the claim follows.
	\end{proof}
	
	Note that S-lemma \cite[\S 2.3]{Yakubovich:1971} enabling the proof still applies if a general quadratic uncertainty set $\mathcal D_{ \rm quad} = \{ \xi \in \mathbb R^{n_\xi} \ | \ \xi^\top V \xi + 2 v^\top \xi \le \rho_s^2 \}$ is considered in Proposition~\ref{pro:SDP}, where $V \in \mathbb R^{n_\xi \times n_\xi}$ is a symmetric matrix (not necessarily positive semi-definite) and $v \in \mathbb R^{n_\xi}$.
	
	We remark that the structured uncertainty set in Example~\ref{Example_1} is tighter than the other uncertainty sets introduced before in modeling the uncertainties on Hankel matrices, and thus it leads to less conservative results when Hankel matrices are used as predictors. However, it requires one to solve a semi-definite program and generally needs more computational effort. As will be discussed in the simulation section, the advantages of the structured uncertainty set come at the cost of more computational effort, and we will also show that though the unstructured uncertainty set is conservative, it can lead to satisfactory performance. One can analogously also  consider $\xi$ on the input/output data to construct Pages matrices and derive $D(g)$ and $c(g)$, though it may result in a very high-dimension semi-definite program, as Page matrices require much longer trajectories.
	
	It is worth noting that for a special case of~\eqref{eq:minmaxDeePC-penalty-rep2}, where $D(g)$ is independent of $g$, the semi-definite programming problem in Proposition~\ref{pro:SDP} can be reduced to a second-order cone program, as shown in the following result. This scenario is useful if the Hankel matrix data has no noise (e.g., the data is generated from a prescribed true model in simulations) and only the initial trajectory $(u_{\rm ini}, y_{\rm ini})$ is corrupted by noise.
	\begin{corollary}
	If $D(g) = D$ in~\eqref{eq:minmaxDeePC-penalty-rep2}, then $g^\star$ is a minimizer of~\eqref{eq:minmaxDeePC-penalty-rep2} if and only if there exists a $(\nu^\star, \lambda^\star)$ such that $g^\star$ also minimizes the second-order cone program
	\begin{equation*}
		\begin{array}{cll}
			\displaystyle \min_{g \in \mathcal G, \nu, \lambda } & \;\;  c(g)^\top c(g) + {\bf 1}^\top \nu + \lambda \rho_s^2 \\
			{\rm s.t. } &  \;\;  \left\| \begin{matrix} (2c(g)^\top DS )_\ell \\
				\nu_\ell + \alpha_\ell - \lambda
			\end{matrix}\right\| \le \nu_\ell - \alpha_\ell + \lambda & \quad \forall \ell \in [H_c],
		\end{array}
	\end{equation*}
	where $S$ is a nonsingular matrix that simultaneously diagonalizes $D^\top D$ and $I_{H_c} \in \mathbb R^{H_c \times H_c}$, namely,
	\[
	S^\top D^\top D S =: {\rm diag} (\alpha) \ , \quad  S^\top S = I_{H_c},
	\]
	and $\alpha = [ \alpha_1, \cdots, \alpha_{H_c} ]$.
	\end{corollary}
	\begin{proof}
	For any fixed $g$, it follows from S-lemma \cite[\S 2.3]{Yakubovich:1971} that the inner maximization problem in~\eqref{eq:minmaxDeePC-penalty-rep2} can be equivalently reformulated as
		\begin{equation*}
			\begin{array}{l}
				\displaystyle \min_{\tau, \lambda}   \;\;  \tau \\
				{\rm s.t. }   \; \begin{bmatrix}
					\tau - \lambda \rho_s^2 - c(g)^\top c(g)  & -c(g)^\top D  \\
					-D^\top c(g) & \lambda I_{H_c} - D^\top D
				\end{bmatrix}  \succeq 0.
			\end{array}
		\end{equation*}
	Since there exists a nonsingular matrix $S$ that simultaneously diagonalizes $D^\top D$ and $I_{H_c} \in \mathbb R^{H_c \times H_c}$, the linear matrix inequality (LMI) constraint above is satisfied if and only if
	\begin{equation}\label{eq:LMI1}
	\begin{bmatrix}
					\tau - \lambda \rho_s^2 - c(g)^\top c(g)  & -c(g)^\top DS  \\
					-(DS)^\top c(g) & \lambda I_{H_c} - {\rm diag} (\alpha)
				\end{bmatrix}   \succeq 0.
	\end{equation}
	Note that such a matrix $S$ always exists.
    It then follows from Schur complement \cite[\S A.5.5]{Boyd:2004} that~\eqref{eq:LMI1} can be reduced to a set of second-order cone constraints.
	\end{proof}

    \section{Robustness Induced by Regularization}
    \label{sec:regl}

   Proposition~\ref{eq:frob} and Corollary~\ref{co:one_norm} show that when unstructured or column-wise uncertainties are considered, the min-max formulation in~\eqref{eq:minmaxDeePC-penalty} reduces to a minimization problem with an additional regularization term on $g$ (see~\eqref{eq:unstructure_1} and \eqref{eq:colwise2}). The robustness induced by different geometry of uncertainties can thus be interpreted as different types of regularization on~$g$. This is consistent with the observation that regularization is instrumental for ensuring good performance when the system is subject to disturbances (see also the discussion in Section~\ref{sec:intro}). 

    Notice that the 2-norm is considered in the objective functions of~\eqref{eq:unstructure_1} and \eqref{eq:colwise2}, which differ from the existing literature where quadratic costs are generally considered (see, e.g., \cite{coulson2019data,berberich2019data,huang2020quad}). For instance, in the spirit of regularized DeePC, one may consider quadratic costs on the input/output signals and a quadratic regularization on $g$ as
    \begin{equation}
		\begin{array}{cl}
			\displaystyle \mathop {{\rm{min}}}\limits_{g \in \mathcal G} & \| \hat U_{\rm F} g \|_R^2 + \| { \hat Y_{\rm F}g  - r} \|_Q^2 + \lambda_u \| \hat U_{\rm P} g - \hat u_{\rm ini} \|^2  \\
			& + \lambda_y \| \hat Y_{\rm P} g - \hat y_{\rm ini} \|^2 + \lambda_g \|g\|^2,
			\label{eq:quad_eqvl}
		\end{array}
	\end{equation}
	or a 1-norm regularization on $g$ as
	\begin{equation}
		\begin{array}{cl}
			\displaystyle \mathop {{\rm{min}}}\limits_{g \in \mathcal G} & \| \hat U_{\rm F} g \|_R^2 + \| { \hat Y_{\rm F}g  - r} \|_Q^2 + \lambda_u \| \hat U_{\rm P} g - \hat u_{\rm ini} \|^2  \\
			& + \lambda_y \| \hat Y_{\rm P} g - \hat y_{\rm ini} \|^2 + \lambda_g \|g\|_1,
			\label{eq:onenorm_eqvl}
		\end{array}
	\end{equation}
	where the set $\mathcal{G}$ is the same as that in~\eqref{eq:minmaxDeePC-penalty} to further robustify the input/output constraints. The following results show how~\eqref{eq:quad_eqvl} and \eqref{eq:onenorm_eqvl} are respectively related to the tractable formulations of robust DeePC~\eqref{eq:unstructure_1} and \eqref{eq:colwise2}, and thus can be considered as special cases of~\eqref{eq:minmaxDeePC-penalty}.

	\begin{theorem}[Robustness induced by quadratic regularization]\label{thm:quad_regl_robust}
	If $g^\star \neq 0$ is a minimizer of~\eqref{eq:quad_eqvl}, then $g^\star$ minimizes \eqref{eq:unstructure} (and equivalently, \eqref{eq:unstructure_1}) with
	\begin{equation}\label{eq:lambda_g}
	\rho_u =  \begin{cases}
		\frac{\lambda_g \sqrt{\left\| g^\star\right\|^2+1}}{\left\|A^{(0)}g^\star - b^{(0)}\right\|} & \text{if $A^{(0)}g^\star \ne b^{(0)}$} \\
		\lambda_g \sqrt{\left\| g^\star\right\|^2+1} & \text{otherwise.}
	\end{cases}
    \end{equation}
    Moreover, $\rho_u$ in \eqref{eq:lambda_g} is strictly monotonically increasing with $\lambda_g$ chosen in \eqref{eq:quad_eqvl}.
	\end{theorem}
	\begin{proof}
	See Appendix~\ref{Appen:A}.
	\end{proof}
	
	Theorem~\ref{thm:quad_regl_robust} shows that by incorporating unstructured uncertainties, robust DeePC reduces to regularized DeePC with quadratic regularization where the input/output constraints are also robustified.
	Note that Theorem~\ref{thm:quad_regl_robust} extends the results in \cite{huang2020quad} by further incorporating robust constraints. It highlights the importance of the regularization from a robust optimization perspective, namely, quadratic regularization of $g$ in \eqref{eq:quad_eqvl} is equivalent to a robust reformulation as in \eqref{eq:unstructure} with an implicit bound for the uncertainty set.
	Hence, by choosing a sufficiently large $\lambda_g$ in~\eqref{eq:quad_eqvl}, a sufficiently large bound $\rho_u$ can be obtained in~\eqref{eq:unstructure} such that Assumption~\ref{assum:D} holds and Theorem~\ref{thm:real_cost} applies.
	
	It has been observed in \cite{coulson2019data} and \cite{dorfler2021bridging} that the 1-norm regularization promotes sparsity, selecting the most informative noisy trajectories to predict the future behaviour. The following result provides new interpretations of 1-norm regularization in DeePC from a min-max optimization perspective.
	
	\begin{theorem}[Robustness induced by 1-norm regularization]
	\label{thm:one_norm}
	If $g^\star \in \mathbb{R}^{H_c}$ is a minimizer of~\eqref{eq:onenorm_eqvl},
	then $g^\star$ minimizes \eqref{eq:column} with $\mathcal D_{\rm gcol} = \{ (\Delta_A,\Delta_b) \ | \ \|(\Delta_A)_i\| \le \rho_c, \forall i \in [H_c],\Delta_b \le \rho_c \}$ (and equivalently, \eqref{eq:colwise2}), where
	\begin{equation}\label{eq:lambda_g_onenorm}
	\rho_c =  \begin{cases}
		\frac{\lambda_g}{2\left\|A^{(0)}g^\star - b^{(0)}\right\|} & \text{if $A^{(0)}g^\star \ne b^{(0)}$} \\
		\lambda_g/2 & \text{otherwise.}
	\end{cases}
    \end{equation}
	Moreover, $\rho_c$ in \eqref{eq:lambda_g_onenorm} is strictly monotonically increasing with $\lambda_g$ chosen in \eqref{eq:onenorm_eqvl}.
	\end{theorem}
	\begin{proof}
	See Appendix~\ref{Appen:B}.
	\end{proof}
	
	Theorem~\ref{thm:one_norm} shows that the 1-norm regularization in \eqref{eq:onenorm_eqvl} provides robustness to column-wise uncertainties in the data matrices and the initial trajectory.
	One can further conclude that a sufficiently large $\lambda_g$ in~\eqref{eq:ReglDeePC_Onenorm} leads to a sufficiently large $\rho_c$ in~\eqref{eq:column} such that Assumption~\ref{assum:D} is satisfied and Theorem~\ref{thm:real_cost} applies.
	The 1-norm regularization on $g$ may be more appropriate than the quadratic regularization when Page matrices or trajectory matrices are used as predictors, because the 1-norm regularization implies a column-wise uncertainty set that assumes no correlation among different columns. Moreover, \eqref{eq:onenorm_eqvl} can be solved by quadratic programming since it can be reformulated as~\eqref{eq:reform Onenorm_Regl_DeePC1} in the proof.

	\section{Simulation Results}
	\label{sec:simulation}
	
	\begin{figure}
\begin{center}
\includegraphics[width=8.6cm]{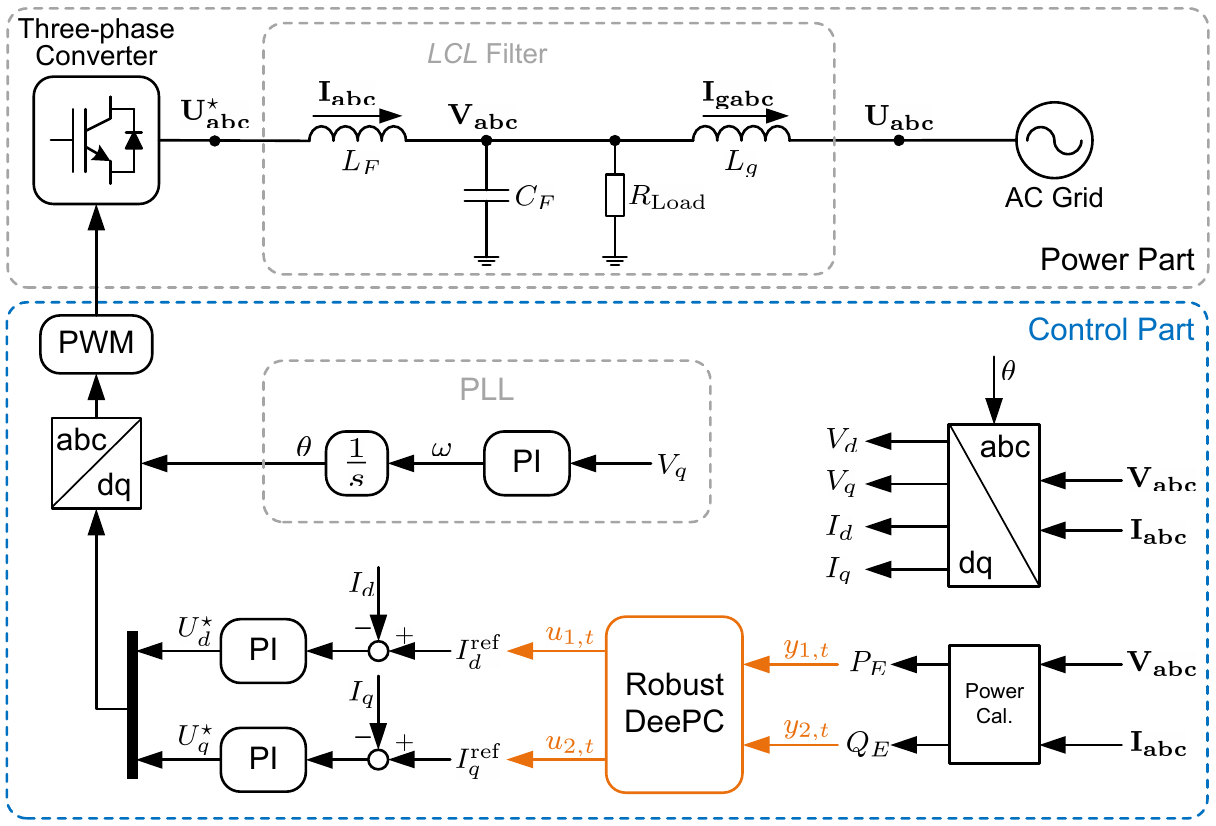}
\vspace{-2mm}
\caption{Control scheme for a grid-connected converter.}
\vspace{0mm}
\label{Fig_Converter}
\end{center}
\end{figure}

    In this section, we provide simulation results to illustrate the effectiveness of the proposed robust DeePC algorithm.

    \subsection{Simulation Case Study}

    We consider a grid-connected three-phase power converter represented as in Fig.~\ref{Fig_Converter}, and apply robust DeePC to regulate the active and reactive power. Conventionally, power regulation of grid-connected converters can be achieved by PI controllers. However, the power grid is ever-changing and in general unknown from the perspective of a converter, which significantly affects the performance of (fixed) PI controllers and may even result in instabilities \cite{huang2019grid}.
    As a remedy, we employ robust DeePC to perform model-free, robust, and optimal power control for converters. Thanks to the data-centric representation, input/output data of the converter can be collected to capture the system dynamics, predict the future behaviors, and calculate optimal control sequences, without assuming a specific model.

    As shown in Fig.~\ref{Fig_Converter}, we choose the active power $P_E = V_dI_d+V_qI_q$ and the reactive power $Q_E = V_qI_d - V_dI_q$ to be the output signals of the converter system, and the robust DeePC algorithm provides optimal control inputs for the current references $I_d^{\rm ref}$, $I_q^{\rm ref}$. Note that $u_{i,t}$ is the $i{\rm th}$ element of $u_t$ and $y_{i,t}$ is the $i{\rm th}$ element of $y_t$ in Fig.\ref{Fig_Converter}. The power regulation is achieved by setting the reference vector in \eqref{eq:minmaxDeePC-penalty} to $r = I_N \otimes {\rm col}(P_0,Q_0)$, where $P_0$ is the active power reference value and $Q_0$ is the reactive power reference value.

    Throughout our simulations, we use the nonlinear converter model; the same trends in the results are, however, also observed if one uses a linearization. We use the base values $f_{\rm b} = 50{\rm Hz}$, $S_{\rm b} = 1.5{\rm kW}$, and $U_{\rm b} = 280{\rm V}$ for per-unit calculations of the converter system. The {LCL} parameters are: $L_F = 0.05{\rm (p.u.)}$ (with resistance $R_F = 0.01{\rm (p.u.)}$), $L_g = 0.05{\rm (p.u.)}$ (with resistance $R_g = 0.01{\rm (p.u.)}$), and capacitor $C_F = 0.05{\rm (p.u.)}$. The local load is $R_{\rm Load} = 2{\rm (p.u.)}$. The PI parameters of the current control are $\{0.2,10\}$. The sampling time for the robust DeePC algorithm is $1{\rm ms}$. The parameters for the robust DeePC are: $T_{\rm ini} = 5$, $N = 25$, $T = 120$, $R = I_{2N}$, $Q = 10^{5}I_{2N}$, $\lambda_u = \lambda_y = 10^5$. The control horizon is $k=N$. 
    In what follows, we assume that the output data (in the Hankel matrices and the initial trajectory) are corrupted by measurement noise, while the input data are known exactly, and we show the realized trajectory of the system outputs (without showing the noise).
    Before the robust DeePC algorithm is activated, persistently exciting white noise signals are injected into the system through $I_d^{\rm ref}$ and $I_d^{\rm ref}$ for $0.12{\rm s}$ to collect the input/output data.  Fig.~\ref{Fig_exciting_signals} plots the input/output responses during this data-collection period, which implies that the active power and the reactive power are perturbed, but within an acceptable range.

\begin{figure}
\begin{center}
\includegraphics[width=8.6cm]{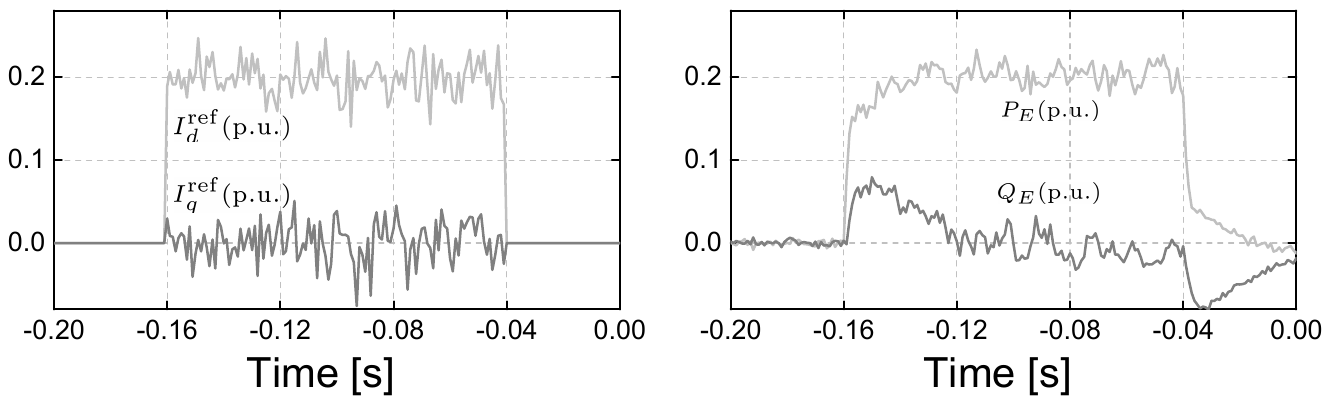}
\vspace{-2mm}
\caption{Input and output responses during the data-collection period.}
\vspace{-2mm}
\label{Fig_exciting_signals}
\end{center}
\end{figure}

    \subsection{Comparison of Conservativeness}

\begin{figure}
\begin{center}
\includegraphics[width=8.0cm]{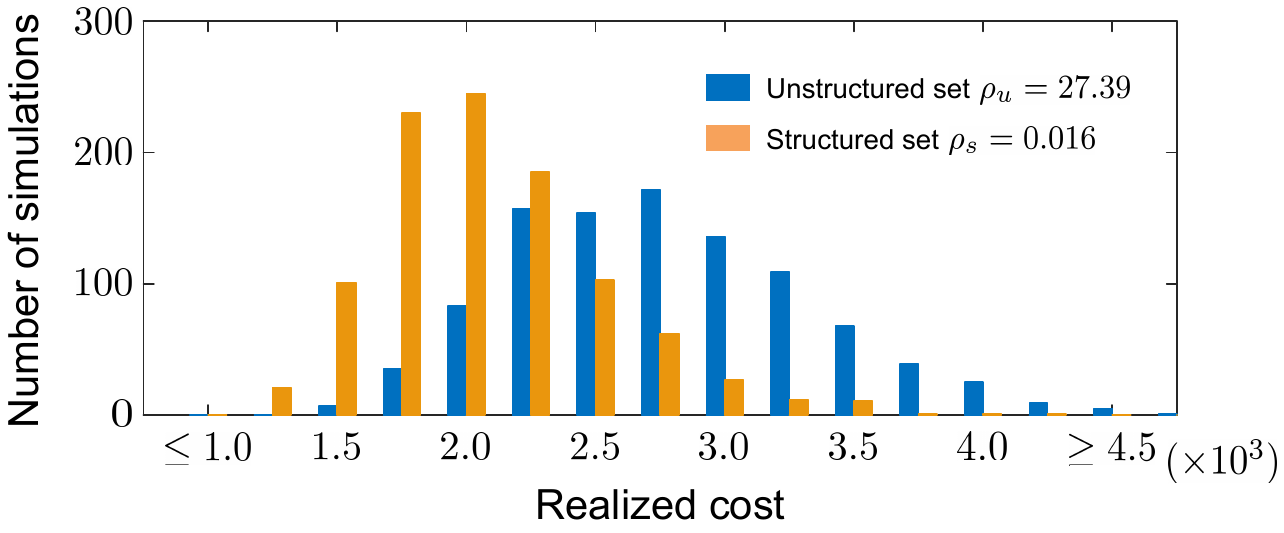}
\vspace{-2mm}
\caption{Comparison of robust DeePC with structured set and unstructured set when the uncertainty samples are drawn from a structured set. Hankel matrices are used as predictors.}
\vspace{-4mm}
\label{Fig_hist1}
\end{center}
\end{figure}

    In this subsection, we compare the performance and conservativeness of robust DeePC when incorporating different uncertainty sets.
    We activate the robust DeePC algorithm at $t=0{\rm s}$ with $P_0 = 0$ and $Q_0 = 0$ and change $P_0$ from $0$ to $0.1{\rm (p.u.)}$ at $t=0.2{\rm s}$. We consider the realized cost of applying the optimal control sequence obtained by solving~\eqref{eq:minmaxDeePC-penalty} at $t=0.2{\rm s}$ in open loop for the whole horizon $N$. Though the realized cost is related to the optimization cost in Theorem~\ref{thm:real_cost}, it is still unclear if a smaller (albeit tight) uncertainty set in robust DeePC leads to a better average realized cost.

    We start with a study where Hankel matrices are used as predictors. Theoretical intuition from Section~\ref{sec:tra} suggests that when Hankel matrices are used, robust DeePC with structured uncertainties should be less conservative than that with unstructured uncertainties. To test this intuition, we assume that the output data $\hat y^{\rm d}$ and $\hat y_{\rm ini}$ are affected by uncertainties residing in the structured set with $\rho_s = 0.016$, and draw 1000 samples that are uniformly-distributed over this structured set. Fig.~\ref{Fig_hist1} plots the corresponding realized costs of applying robust DeePC with the structured set (by solving~\eqref{eq:sdp} where $\rho_s=0.016$) and the unstructured set (by solving~\eqref{eq:unstructure_1} where $\rho_u=27.39$, the smallest value such that the unstructured set contains the structured set), respectively.
    Fig.~\ref{Fig_hist1} shows that robust DeePC with structured set performs better than that with unstructured set, which confirms our intuition.
    Realized cost becomes even higher with the column-wise and interval uncertainty sets when scaled to contain the structured uncertainty set (data not shown).

\begin{figure}
\begin{center}
\includegraphics[width=8.0cm]{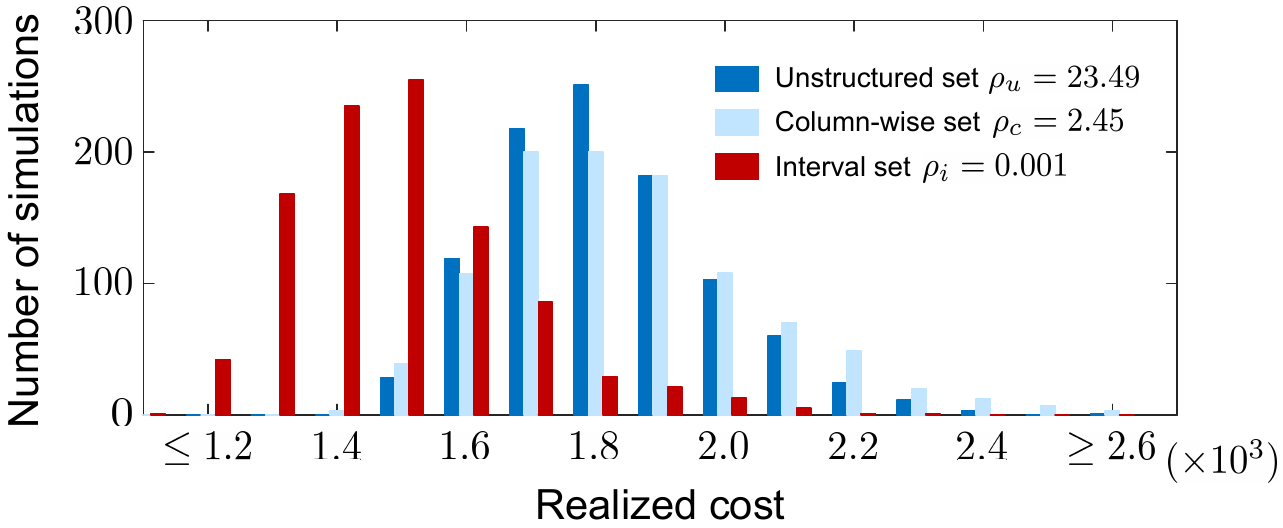}
\vspace{-2mm}
\caption{Comparison of robust DeePC with interval set, column-wise set, and unstructured set when the uncertainty samples are drawn from a interval set. Page matrices are used as predictors.}
\vspace{-4mm}
\label{Fig_hist2}
\end{center}
\end{figure}

    The analysis in Section~\ref{sec:tra} indicates that column-wise set and interval set are more appropriate choices when Page matrices are used.
    To validate this, we next consider Page matrices as predictors and compare the performance of unstructured set, column-wise set, and interval set. We collect a longer input/output trajectory such that the constructed Page matrices have the same dimensions as the Hankel matrices. We first assume that the output data is affected by uncertainties residing in an interval set with the upper and lower bounds for the output signals being $\pm \rho_i$ ($\rho_i=0.001$).
    We consider 1000 samples uniformly-distributed over this interval set and plot in Fig.~\ref{Fig_hist2} the corresponding realized costs of applying robust DeePC with respectively the same interval set (by solving~\eqref{eq:interval_QP}), the column-wise set (by solving~\eqref{eq:colwise2} where for all $k \in [H_c]$: $(\rho_A)_k = \rho_c=2.45$), and the unstructured set (by solving~\eqref{eq:unstructure_1} where $\rho_u=23.49$); we again choose the smallest $\rho_c$ and $\rho_u$ such that the considered interval set is contained in the column-wise set and the unstructured set.
    Fig.~\ref{Fig_hist2} shows that robust DeePC with the interval set achieves better average and worst-case performance than the other two methods, because the interval set is tight in this case. However, the cases with the column-wise set did not show superior performance over the unstructured set, as the column-wise set is also conservative under the above setting.

\begin{figure}
\begin{center}
\includegraphics[width=8.0cm]{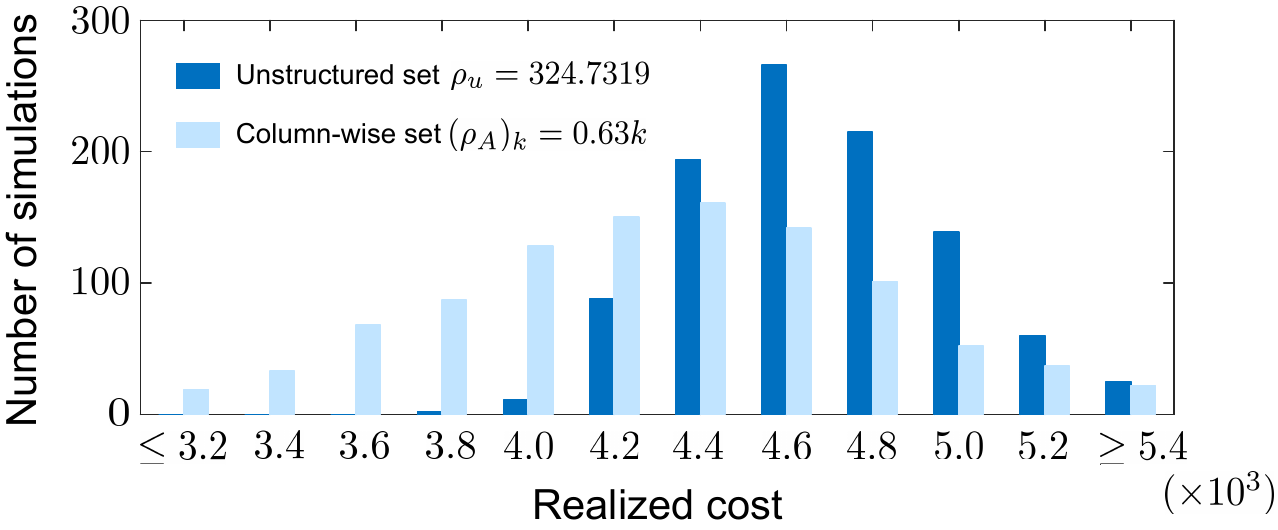}
\vspace{-2mm}
\caption{Comparison of robust DeePC with column-wise set and unstructured set when the uncertainty samples are drawn from a column-wise set. Trajectory matrices are used as predictors.}
\vspace{-4mm}
\label{Fig_hist3}
\end{center}
\end{figure}

    One advantage of incorporating column-wise uncertainties in robust DeePC is that one can compactly consider different bounds for different columns in the data matrices. To demonstrate this, we use trajectory matrices as predictors and test the performance of robust DeePC when the columns in the trajectory matrices are subject to different levels of uncertainties, which could be the case when the columns are from different experiments. We consider 1000 uncertainty samples uniformly-distributed over the column-wise set~\eqref{eq:column_wise_set} where for all $k \in [H_c]$: $(\rho_A)_k = 0.63k$ and $\rho_b = 0.63H_c$. Fig.~\ref{Fig_hist3} shows the performance of robust DeePC with respectively the considered column-wise set (by solving~\eqref{eq:colwise2}) and the unstructured set that contains the column-wise set (by solving~\eqref{eq:unstructure_1} where $\rho_u=324.73$). Under this setting, we observe that robust DeePC with the column-wise set has a better average performance than that with the unstructured set, thanks to the tighter uncertainty representation.

    The average runtime of these simulations (on an Intel Core i7~9750H CPU with 16GB~RAM) is shown in Table~\ref{tab:runtime}.
    As expected, it takes significantly more time to solve the semi-definite program in \eqref{eq:sdp}.

    \begin{table}
        \renewcommand\arraystretch{1.6}
		\centering
		\caption{Average runtime when applying different uncertainty sets.}
		\label{tab:runtime}
		\begin{tabular}{c|cc}
			\hline
			Uncertainty sets  & Tractable formulations & Average runtime  \\ \hline
			Unstructured set & Eq.~\eqref{eq:unstructure_1} & $0.0146{\rm s}$ \\
			Column-wise set & Eq.~\eqref{eq:colwise2} & $0.0269{\rm s}$ \\
			Interval set
			& Eq.~\eqref{eq:interval_QP} with \eqref{eq:interval_bound} & $0.1063{\rm s}$ \\
			Structured set
			& Eq.~\eqref{eq:sdp} with \eqref{eq:sdp_Dc} & $4.3234{\rm s}$ \\ \hline
		\end{tabular}
		\end{table}

    \subsection{Comparison of On-line Performance}

    In the previous section, we consider uncertainty samples that are uniformly-distributed over some prescribed uncertainty set to test the conservativeness of robust DeePC. However, in practice, measurement noise may be better represented by white noise.
    Hence, for a more realistic comparison, we then turn to band-limited Gaussian white noise~\cite{whitenoise} in the output measurements and test the on-line performance of robust DeePC under two different noise levels: 1) noise power: $4\times 10^{-10} {\rm (p.u.)}$; 2) noise power: $1.6\times 10^{-7} {\rm (p.u.)}$. Two example sequences are given in Fig.~\ref{Fig_noise_examples} to illustrate the range of uncertainties with different noise power.

    In what follows, we consider Hankel matrices as predictors. After the robust DeePC algorithm is activated at $t = 0{\rm s}$ with $P_0 = 0$ and $Q_0 = 0$, we change $P_0$ from $0$ to $0.1{\rm (p.u.)}$ at $t = 0.2{\rm s}$ and back to $0$ at $t = 0.4{\rm s}$. Fig.~\ref{Fig_hist4}~(a) plots the active power responses of applying robust DeePC with different uncertainty sets, where the noise power is $4\times 10^{-10} {\rm (p.u.)}$. For the interval set, we consider the upper and lower bounds of the output signals to be $\pm 0.001$ (compare Fig.~\ref{Fig_noise_examples}~(a)), then choose the smallest $\rho_c$ such that the interval set is contained in the column-wise set; we choose the smallest $\rho_s$ such that the structured set contains the interval uncertainty $[-0.001,0.001]$ in the outputs; finally, we choose the smallest $\rho_u$ such that the unstructured set contains the structured set.
    In all cases, the converter has satisfactory tracking performance. Moreover, the performance is improved by reducing the control horizon from $25$ to $5$, thanks to the faster feedback.

    The above simulations with $k=25$ (i.e., the whole input sequence is applied, in line with Theorem~\ref{thm:real_cost}) are repeated 100 times with different data sets to construct the Hankel matrices and different random seeds to generate the measurement noise. The histogram in Fig.~\ref{Fig_hist4}~(b) shows the realized input/output cost from $0$ to $0.6{\rm s}$ (i.e., $\sum_{i = 0}^{600} \| u_{{\rm sim},i} \|^2_R + \| y_{{\rm sim},i} - r_{{\rm sim},i} \|^2_Q$ where $u_{{\rm sim},i}$, $y_{{\rm sim},i}$, and $r_{{\rm sim},i}$ are respectively the input, output, and reference data at time $i$ obtained from the simulations). All the cases have satisfactory performance, and the cases with the structured set and the interval set have better average performance than the cases with the unstructured set or the column-wise set, consistent with the analysis in Section~\ref{sec:tra}.

\begin{figure}
\begin{center}
\includegraphics[width=8.6cm]{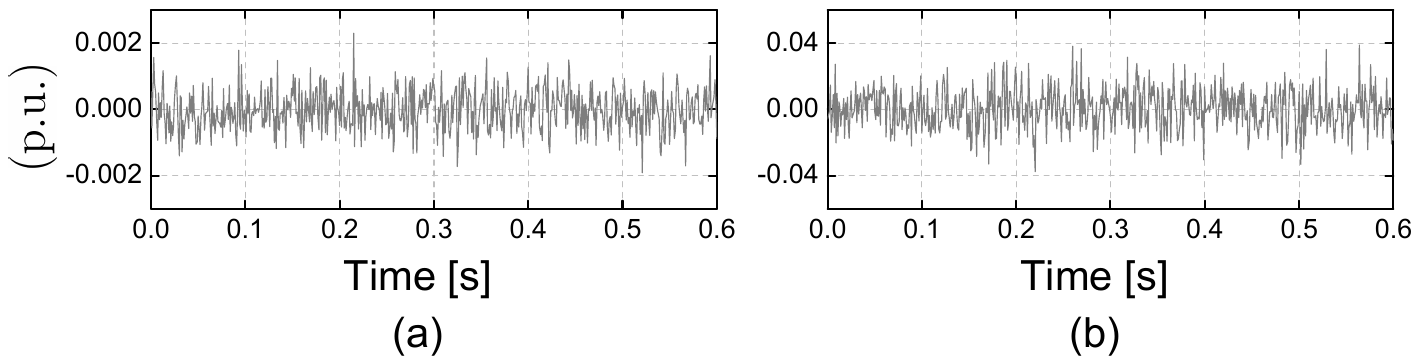}
\vspace{-2mm}
\caption{Two example sequences of white noise signals. (a) Noise power: $4\times 10^{-10} {\rm (p.u.)}$. (b) noise power: $1.6\times 10^{-7} {\rm (p.u.)}$.}
\vspace{-4mm}
\label{Fig_noise_examples}
\end{center}
\end{figure}

\begin{figure}
\begin{center}
\includegraphics[width=8.6cm]{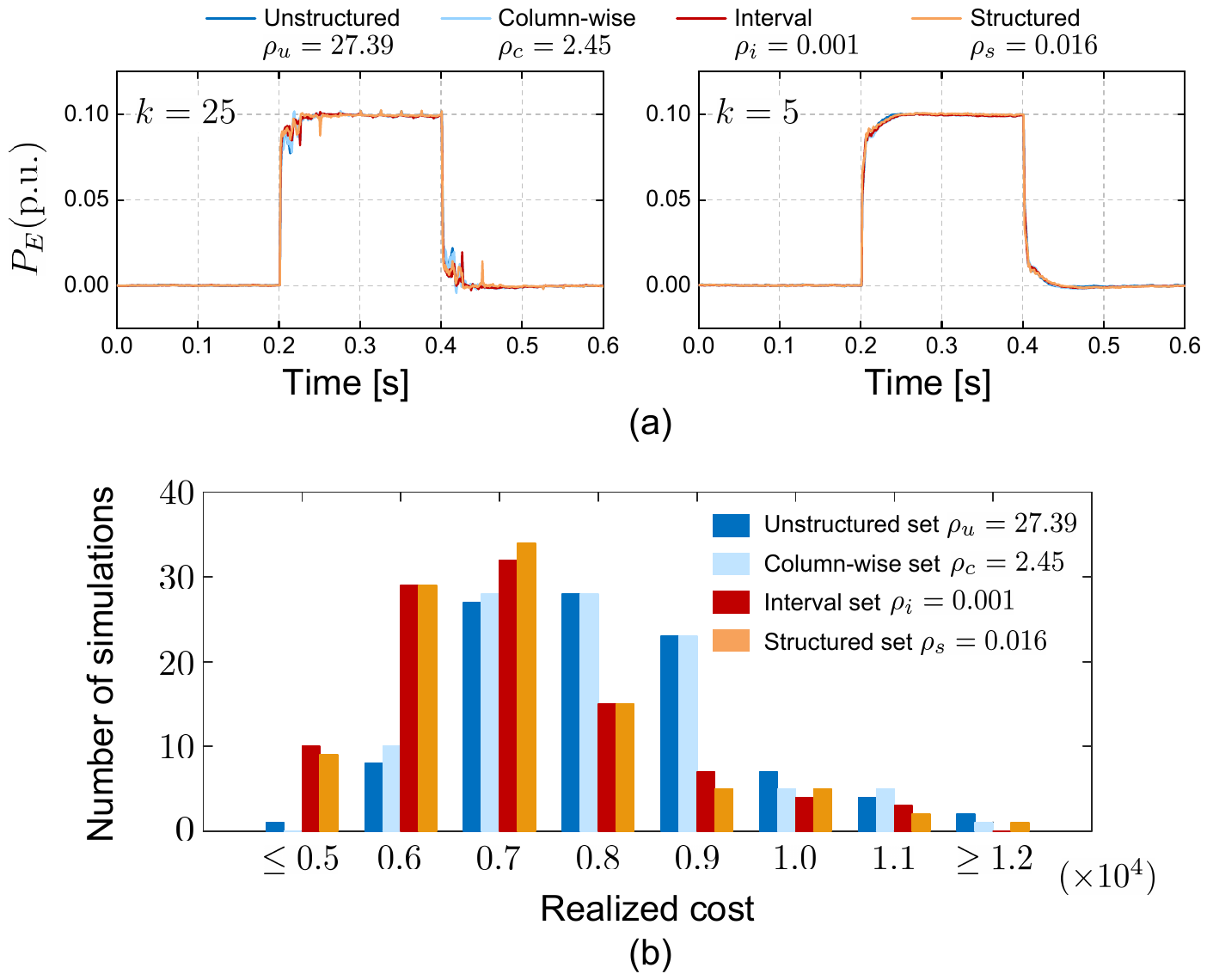}
\vspace{-2mm}
\caption{Comparison of on-line performance with noise power being $4\times 10^{-10}{\rm (p.u.)}$. (a) Time-domain responses. (b) Realized costs.}
\vspace{0mm}
\label{Fig_hist4}
\end{center}
\end{figure}

\begin{figure}
\begin{center}
\includegraphics[width=8.6cm]{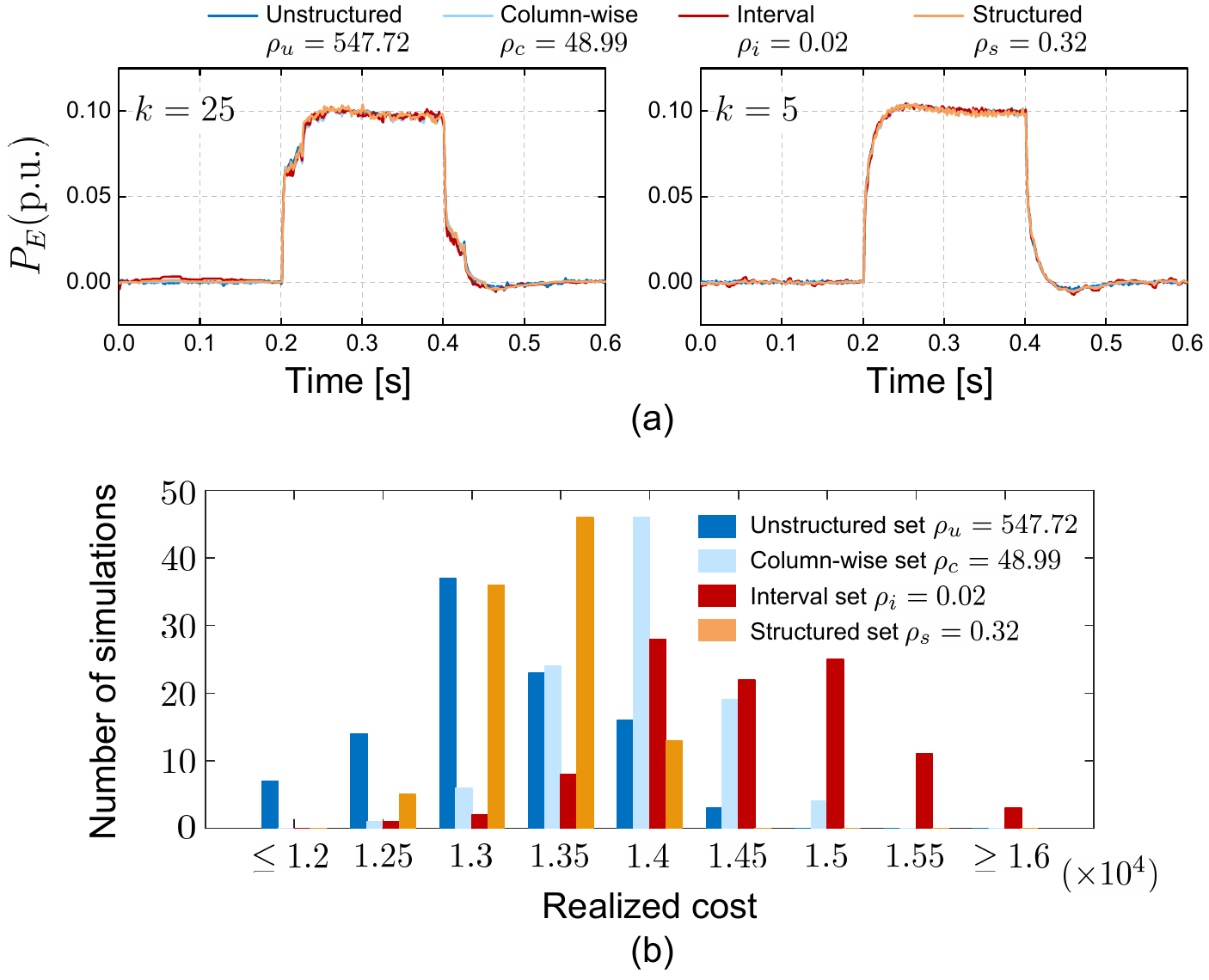}
\vspace{-2mm}
\caption{Comparison of on-line performance with noise power being $4\times 1.6^{-7}{\rm (p.u.)}$. (a) Time-domain responses. (b) Realized costs.}
\vspace{-3mm}
\label{Fig_hist5}
\end{center}
\end{figure}

    We next test the performance of robust DeePC when a higher measurement noise level is considered. Fig.~\ref{Fig_hist5} plots the active power responses when the noise power is $4\times 1.6^{-7}{\rm (p.u.)}$, with the bounds for the uncertainty sets scaled accordingly. For instance, the upper and lower bounds of the output signals are $\pm 0.02$ in the interval set (compare Fig.~\ref{Fig_noise_examples}~(b)). All the cases still have satisfactory tracking performances, but are worse than those in Fig.~\ref{Fig_hist4}~(a) as the impact of noise increases. The above simulations with $k=25$ are repeated 100 times, and the histogram in Fig.~\ref{Fig_hist5}~(b) compares the performance of robust DeePC with different uncertainty sets. Surprisingly, the cases with the unstructured set have the best average performance.
    This counter-intuitive effect is due to the fact that the noise is not uniformly-distributed over any of the uncertainty sets, and it suggests that robust DeePC with unstructured set may have superior performance in practice.

\subsection{Resilience in Presence of Outliers}

To test the resilience of robust DeePC, we consider bad data points in the output trajectory $\hat y^{\rm d}$ used to construct the Hankel matrices. The output data is set as 0 in these bad data points, to reflect, for example, the loss of a measurement. The histogram in Fig.~\ref{Fig_hist6} displays the realized input/output cost from $0$ to $0.6{\rm s}$ in the presence of $5$ such bad data points (appeared in the middle of the data sequence of $P_E$). The robust DeePC with the structured set has the best average performance in this case, which we attribute to the Hankel structure of the uncertainties. Robust DeePC also leads to satisfactory performance when the interval set or the column-wise set are considered. By comparison, the performance degrades significantly when the unstructured set is used, suggesting that this may result in unacceptable resilience. This problem can possibly be resolved by setting a sufficiently large bound, but at the cost of being overly conservative in normal situations.

\begin{figure}
\begin{center}
\includegraphics[width=7.7cm]{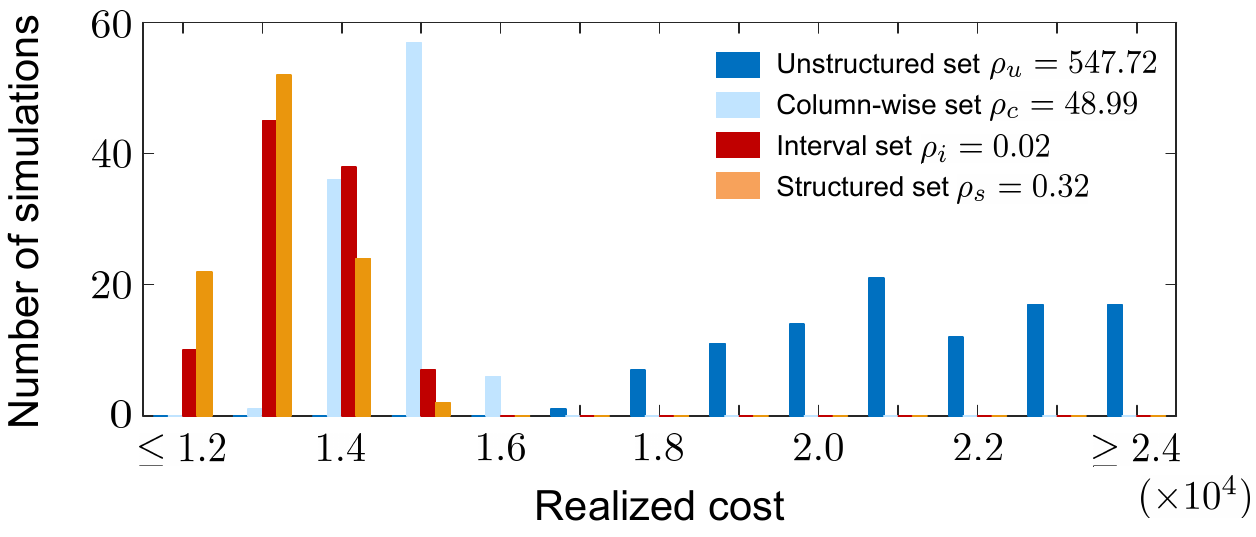}
\vspace{-2mm}
\caption{Comparison of on-line performance with $5$ bad data points.}
\vspace{-3mm}
\label{Fig_hist6}
\end{center}
\end{figure}

\section{Conclusions}
	This paper proposed a robust DeePC framework to perform robust and optimal model-free predictive control. The robust DeePC involves solving a min-max optimization problem to robustify the optimal control sequence against uncertainties in the input/output data that used for predictions. We showed that by applying robust DeePC, the realized input/output cost can be bounded if the considered uncertainty set captures the perfect input/output data. We explicitly derived tractable formulations for robust DeePC when different geometries of uncertainty sets are incorporated. In particular, when Hankel matrices are used as predictors, we illustrated how uncertainties with Hankel structures can be taken into account in a structured uncertainty set to reduce the conservativeness. By incorporating appropriate geometries of uncertainty sets, the robust DeePC recovers existing regularized DeePC algorithms with robustified constraints. In particular, a quadratic regularization corresponds to considering an unstructured uncertainty set, while a 1-norm regularization corresponds to a column-wise uncertainty set. The robust DeePC algorithm with different uncertainty sets was tested by high-fidelity simulations on a grid-connected converter system, which shows satisfactory performance even with noisy measurements and bad data.

	\appendices
	
	\section{Proof of Theorem~\ref{thm:quad_regl_robust}} \label{Appen:A}
	
	We start by compactly rewriting~\eqref{eq:quad_eqvl} as
	\begin{equation}\label{eq:quad_eqvl_1}
	\mathop{\min}\limits_{g \in \mathcal G}\;
	\| A^{(0)}g - b^{(0)} \|^2 + \lambda_g \|g\|^2\,.
	\end{equation}
	Note that the set $\mathcal G$ admits a tractable reformulation as the semi-infinite constraints can be reformulated into a finite number of convex quadratic inequalities using standard robust optimization techniques, as shown in Table~\ref{tab:g.examples}. {Accordingly, we compactly rewrite $\mathcal{G}$ as $\mathcal{G} = \{g \ | \ Gg+ \mathcal{N}_V(g) \le q\}$ where $\mathcal{N}_V(g) = [\|V^{(1)}g\|\; \|V^{(2)}g\|\; \cdots \; \|V^{(n_q)}g\|]^\top$ for some $G \in \mathbb{R}^{n_q \times H_c}$, $v,q \in \mathbb{R}^{n_q}$, and $V^{(i)} \in \mathbb{R}^{n_q \times n_q}\; \forall i \in [n_q]$.}
	
	Consider the Lagrangian of \eqref{eq:quad_eqvl_1}
	\begin{equation}
	\mathcal{L}_Q(g,\mu_Q) = \|A^{(0)}g-b^{(0)}\|^2 + \lambda_g \|g\|^2 + \mu_Q^\top (Gg + \mathcal{N}_V(g) - q),
	\end{equation}
	where $\mu_Q$ is the vector of the dual variables. Since $\mathcal G$ is nonempty, there exist a solution $(g^\star,\mu_Q^\star)$ to the Karush–Kuhn–Tucker (KKT) conditions of \eqref{eq:quad_eqvl_1}
	\begin{subequations}
    \begin{empheq}[left={\empheqlbrace}]{align}
	& \begin{array}{rl}
	\hspace{-1.6mm} 2{(A^{(0)})^\top (A^{(0)}g-b^{(0)})} + 2{\lambda_g g} +   G^\top \mu_Q & \\
	+ \sum_{i \in [n_q]}(\mu_Q)_i  \frac{(V^{(i)})^\top V^{(i)}g}{\|V^{(i)}g\|} & \hspace{-2.5mm} = 0  \,,
	\end{array}           \label{eq:KKT_LQ1}\\
	& \mu_Q^\top (Gg + \mathcal{N}_V(g) - q) = 0  \,,   \label{eq:KKT_LQ2}  \\
	& Gg + \mathcal{N}_V(g) \le  q            \,,     \label{eq:KKT_LQ3}  \\
	& \mu_Q \ge 0          \,.       \label{eq:KKT_LQ4}
    \end{empheq}
    \end{subequations}
    Therefore, the vector $g^\star$ is a minimizer of \eqref{eq:quad_eqvl_1}.
	
	Following Proposition~\ref{eq:frob}, the min-max problem~\eqref{eq:unstructure} can be equivalently reformulated as~\eqref{eq:unstructure_1}.
Consider the Lagrangian of \eqref{eq:unstructure_1}
\begin{equation}\label{eq:L}
\mathcal{L}(g,\mu) {\hspace{-0.7mm}} = {\hspace{-0.7mm}} \|A^{(0)}g-b^{(0)}\| + {\rho_u}\sqrt{\left\| g \right\|^2+1} + \mu^\top(Gg + \mathcal{N}_V(g) - q),
\end{equation}
where $\mu$ is the vector of the dual variables. By choosing $\rho_u$ in \eqref{eq:lambda_g}, it can be verified that $(g^\star,  \mu^\star, y^\star)$, where
\begin{equation*}
	(\mu^\star,  y^\star)  = \begin{cases}
	\hspace{-1.75mm}  \begin{array}{r}
	\left( \frac{\mu_Q^\star}{2 \| A^{(0)}g^\star - b^{(0)} \|}, \frac{(A^{(0)})^\top (A^{(0)} g^\star - b^{(0)})}{\| A^{(0)} g^\star - b^{(0)} \|} \right)  \\
	\text{if $A^{(0)}g^\star \ne b^{(0)}$,}
	\end{array} \\
	\left(\frac{\mu_Q^\star}{2} , 0 \right) \quad  \text{otherwise,} \end{cases}
\end{equation*}
and $g^\star$ as before,
satisfy the KKT conditions of \eqref{eq:L}:
\begin{subequations}
\begin{empheq}[left={\empheqlbrace}]{align}
& \begin{array}{ll}
\hspace{-1.6mm} y + \frac{\rho_u g}{\sqrt{\|g\|^2+1}} +   G^\top \mu  \\
\hspace{+1.0mm} + \sum_{i \in [n_q]}\mu_i  \frac{(V^{(i)})^\top V^{(i)}g}{\|V^{(i)}g\|}= 0 \,,
\end{array}    \label{eq:KKT_L1}\\
& \hspace{-2mm} \begin{array}{r}\|A^{(0)} \bar g - b^{(0)} \| \ge \|A^{(0)} g - b^{(0)} \|  + y^\top (\bar g - g) \\ \forall \bar g \in \mathbb R^{H_c} \,, \end{array} \\
& \mu^\top (Gg + \mathcal{N}_V(g) - q) = 0  \,,   \label{eq:KKT_L2}  \\
& Gg + \mathcal{N}_V(g) \le  q            \,,     \label{eq:KKT_L3}  \\
& \mu \ge 0            \,.     \label{eq:KKT_L4}
\end{empheq}
\end{subequations}
Thus, the vector $g^\star$ is also a minimizer of \eqref{eq:unstructure_1} (and \eqref{eq:unstructure}).
	
	Next we prove the monotonic relationship between $\rho_u$ and~$\lambda_g$. Let $g_1$ be the minimizer of \eqref{eq:quad_eqvl_1} with $\lambda_g = \lambda_{g1} >0$ (and the minimizer of \eqref{eq:unstructure_1} with $\rho_u = \rho_1$), and let $g_2$ be the minimizer of \eqref{eq:quad_eqvl_1} with $\lambda_g = \lambda_{g2} > \lambda_{g1}$ (and the minimizer of \eqref{eq:unstructure_1} with $\rho_u = \rho_2$). If $g_1 \neq g_2$, according to the definitions of $g_1$ and $g_2$, we have
\begin{equation*}
\|A^{(0)}g_1-b^{(0)}\|^2 + \lambda_{g1}{\left\| g_1 \right\|^2} - \lambda_{g1}{\left\| g_2 \right\|^2} < \|A^{(0)}g_2-b^{(0)}\|^2 \,,
\end{equation*}
\begin{equation*}
\|A^{(0)}g_2-b^{(0)}\|^2 < \|A^{(0)}g_1-b^{(0)}\|^2 + \lambda_{g2}{\left\| g_1 \right\|^2} - \lambda_{g2}{\left\| g_2 \right\|^2} \,,
\end{equation*}
leading to
\begin{equation*}
\lambda_{g1}({\left\| g_1 \right\|^2} - {\left\| g_2 \right\|^2}) < \lambda_{g2}({\left\| g_1 \right\|^2} - {\left\| g_2 \right\|^2}) \,,
\end{equation*}
which indicates that ${\left\| g_1 \right\|^2} > {\left\| g_2 \right\|^2}$ because $ 0 < \lambda_{g1} < \lambda_{g2}$. Then, we have ${\left\| g_1 \right\|} > {\left\| g_2 \right\|}$. Since $g_1$ minimizes~\eqref{eq:unstructure_1} with $\rho = \rho_1$ and $g_2$ minimizes~\eqref{eq:unstructure_1} with $\rho = \rho_2$, we also have
\begin{equation*}
\begin{split}
&\|A^{(0)}g_1-b^{(0)}\| + \rho_1\sqrt{\left\| g_1 \right\|^2+1} - \rho_1\sqrt{\left\| g_2 \right\|^2+1} \\
<& \|A^{(0)}g_2-b^{(0)}\| \\
< &\|A^{(0)}g_1-b^{(0)}\| + \rho_2\sqrt{\left\| g_1 \right\|^2+1} - \rho_2\sqrt{\left\| g_2 \right\|^2+1} \,.
\end{split}
\end{equation*}
It can then be deduced that $\rho_1 < \rho_2$. If $g_1 = g_2$, we have $\rho_1 < \rho_2$ according to \eqref{eq:lambda_g}. Hence, $\rho_u$ is increasing with the increase of $\lambda_g$ if $g^\star \neq 0$. This completes the proof.

\section{Proof of Theorem~\ref{thm:one_norm}}\label{Appen:B}
	Notice that \eqref{eq:onenorm_eqvl} can be compactly rewritten as
	\begin{equation}\label{eq:ReglDeePC_Onenorm}
	\mathop{\min}\limits_{g \in \mathcal G}\;
	\| A^{(0)}g - b^{(0)} \|^2 + \lambda_g \|g\|_1\,.
	\end{equation}
	When $(\rho_A)_i = \rho_b = \rho_c \in \mathbb{R}_{>0} \ \forall i$, Corollary~\ref{co:one_norm} implies that $g^\star$ is a minimizer of \eqref{eq:column} if and only if $g^\star$ is a minimizer of the following conic quadratic optimization problem
	\begin{equation}\label{eq:ReglDeePC_Onenorm_conic}
	\mathop{\min}\limits_{g \in \mathcal G}\;
	\| A^{(0)}g - b^{(0)} \| + \rho_c \|g\|_1\,.
	\end{equation}

	In what follows, we show that If $g^\star$ is a minimizer of \eqref{eq:ReglDeePC_Onenorm}, then $g^\star$ also minimizes \eqref{eq:ReglDeePC_Onenorm_conic} by choosing $\rho_c$ according to \eqref{eq:lambda_g_onenorm}. To this end, first observe that both \eqref{eq:ReglDeePC_Onenorm} and \eqref{eq:ReglDeePC_Onenorm_conic} have the same feasible region, and they can be equivalently reformulated into the following compact forms
	\begin{align}
    &\min_{\left(g, \nu \right) \in \mathcal G' }  \;\; \| A^{(0)}g-b^{(0)}  \|^2 + \lambda_g {\bf 1}_{H_c}^\top  \nu, & {\rm and}
    \label{eq:reform Onenorm_Regl_DeePC1}\\
    &\min_{\left(g, \nu \right) \in \mathcal G' }  \;\; \| A^{(0)}g-b^{(0)}  \| + \rho_c {\bf 1}_{H_c}^\top  \nu , & {\rm respectively,}
    \label{eq:reform convex quadratic one norm}
	\end{align}
	where $\mathcal G' = \{ (g, \nu ) \ | \ g \in \mathcal G, \ -\nu \le g \le \nu \}$ is a nonempty polyhedron. Note that the set $\mathcal G$ can be compactly rewritten as $\mathcal G = \{ g \ | \ G_1g + V_1 |g| \le q_1 \}$ for some $G_1, V_1 \in \mathbb{R}^{n_{q1}\times H_c}$, and $q_1 \in \mathbb{R}^{n_{q1}}$. Then, one can replace $|g|$ by $\nu$ such that the set $\mathcal G'$ can be compactly rewritten as $\mathcal G' = \{ (g,\nu) \ | \ G_\nu {\rm col}(g,\nu) \le q_\nu \}$ for some $G_\nu \in \mathbb{R}^{n_{q_\nu}\times H_c}$ and $q_\nu \in \mathbb{R}^{n_{q_\nu}}$.
	
	Then, analogous to the proof of Theorem~\ref{thm:quad_regl_robust}, one can derive the KKT conditions of~\eqref{eq:reform Onenorm_Regl_DeePC1} and~\eqref{eq:reform convex quadratic one norm}, and then check that by choosing $\rho_c$ according to~\eqref{eq:lambda_g_onenorm}, if $g^\star$ is a minimizer of~\eqref{eq:reform Onenorm_Regl_DeePC1}, then it is also a minimizer of~\eqref{eq:reform convex quadratic one norm}.

	It remains to prove the monotonic relationship between $\rho_c$ and $\lambda_g$. Note that~\eqref{eq:ReglDeePC_Onenorm} and~\eqref{eq:ReglDeePC_Onenorm_conic} have unique solutions thanks to~\cite[Lemmas~3~and~4]{tibshirani2013lasso}, which apply because $A^{(0)}$ contains noisy data so that each element of $A^{(0)}$ can be seen as a random variable drawn from a continuous distribution. Based on this fact, one can analogously prove the monotonic relationship by referring to the proof of Theorem~\ref{thm:quad_regl_robust}.

	\normalem
    \bibliographystyle{IEEEtran}
	\bibliography{references}

\end{document}